\pdfoutput=1

\documentclass[pdflatex,sn-mathphys-num,iicol]{sn-jnl}%

\usepackage{amsmath,amssymb,amsfonts}
\usepackage{mathrsfs, amsthm}
\usepackage{booktabs}
\usepackage{bm}
\usepackage{mathtools}
\usepackage{array}
\usepackage{algorithm}
\usepackage{algorithmicx}
\usepackage{algpseudocode}
\usepackage{enumitem}
\usepackage{subcaption}
\usepackage{soul}
\usepackage{graphicx}

\usepackage{tikz}
\usetikzlibrary{arrows.meta,positioning,calc,shapes.geometric,fit,backgrounds}

\theoremstyle{thmstyleone}
\newtheorem{theorem}{Theorem}

\newtheorem{corollary}{Corollary}
\newtheorem{proposition}{Proposition}

\theoremstyle{thmstyletwo}
\newtheorem{remark}{Remark}

\newcommand{\bq}{\bm{q}}
\newcommand{\bqo}{\bm{q}_{\mathrm{obs}}}
\newcommand{\bqh}{\bm{q}_{\mathrm{hid}}}
\newcommand{\bxi}{\boldsymbol{\xi}}
\newcommand{\btheta}{\boldsymbol{\theta}}
\newcommand{\bPhi}{\boldsymbol{\Phi}}

\newcommand{\calL}{\mathcal{L}}
\newcommand{\fphys}{f_{\mathrm{phys}}}

\newcommand{\R}{\mathbb{R}}
\newcommand{\norm}[1]{\left\lVert #1 \right\rVert}
\newcommand{\diag}{\mathrm{diag}}

\newcommand{\OO}{\Omega}
\newcommand{\bL}{\bm{L}}
\newcommand{\br}{\bm{r}}
\newcommand{\dbqo}{\dot{\bm{q}}_{\mathrm{obs}}}
\newcommand{\ddbqo}{\ddot{\bm{q}}_{\mathrm{obs}}}
\newcommand{\dbqh}{\dot{\bm{q}}_{\mathrm{hid}}}

\newcommand{\tbqo}{\tilde{\bm{q}}_{\mathrm{obs}}}
\newcommand{\dtbqo}{\dot{\tilde{\bm{q}}}_{\mathrm{obs}}}
\newcommand{\ddtbqo}{\ddot{\tilde{\bm{q}}}_{\mathrm{obs}}}

\newcommand{\tablefont}{\scriptsize}

\begin{document}

\title[SPIRAL-PO: Partially Observed Identification of Rotating Machinery]%
      {SPIRAL-PO: Symbolic Identification of Partially Observed
       Nonlinear Dynamics with Application to Rotating Machinery}

\author*[1]{\fnm{Mohammad~A.} \sur{Ayoubi}}\email{maayoubi@scu.edu}

\affil*[1]{\orgdiv{Department of Mechanical Engineering},
           \orgname{Santa Clara University},
           \orgaddress{\city{Santa Clara},
                       \state{CA},
                       \postcode{95053},
                       \country{USA}}}
                       
\author[2]{\fnm{Sina} \sur{Piramoon}}\email{spiramoon@scu.edu}

\affil[2]{\orgdiv{Laboratory Products Division},
           \orgname{Thermo Fisher Scientific},
           \orgaddress{\city{Santa Clara},
                       \state{CA},
                       \postcode{95054},
                       \country{USA}}}

\abstract{%
Many engineering systems contain states that cannot be measured
directly, yet these hidden states act on the measured outputs. We
present SPIRAL-PO (Symbolic Physics-Informed Residual Augmentation
Loop --- Partially Observed), a framework for identifying an
interpretable \emph{observable-space closure}: a model, written in the
measured coordinates, of how the unmeasured states act on the measured
dynamics. A physics seed is augmented by candidate terms from a
mechanism-labeled library; terms are selected by statistical tests
designed for numerically differentiated, serially correlated data,
using a weak (integral) formulation, and are checked on independent
records. For an idealized regression we give sufficient conditions for
recovering the closure coefficients, a rank condition under which no
estimator can recover them, and a per-term detection bound; these
results apply only to systems that admit an instantaneous observable
closure. On a simulated four-degree-of-freedom vertical rotor with
cubic bearing stiffness, in which only the lateral displacements are
measured, the method recovers the gyroscopic and cubic-stiffness
signatures of the hidden tilt dynamics in every noise realization at
2\% noise; with a mechanism removed it never admits the cubic term and
admits the gyroscopic term in 15--25\% of realizations, above the
nominal level, because the closure is only partly instantaneous. A
residual network and a quasi-random projection stage add no measurable
benefit, and pointwise tests that ignore the serial correlation of
differentiated data admit absent mechanisms more often. On a laboratory rotor with simultaneous $x,y$
measurement and an order-tracked speed, the method detects only a
speed-dependent stiffening; gyroscopic and cubic terms obtained with an
assumed constant-deceleration speed are consistent with the speed
error, which produces such terms in simulation.}


\maketitle

\IfFileExists{GraphicalAbstract.png}{%
\begin{figure*}[!t]
  \centering
  \includegraphics[width=0.9\textwidth]{GraphicalAbstract.png}
  \caption*{Graphical abstract}
\end{figure*}}{}

\section{Introduction}
\label{sec:intro}

A recurring challenge in engineering dynamics is that the models we
need are governed by states we cannot measure. Tilt angles in rotating
machines, internal flow variables in fluid systems, and aeroelastic
modes in flexible structures all exert observable effects---they couple
into the measured outputs---yet direct sensing is impractical or
impossible. The governing equations of the observable channels then
contain structured, unknown terms that look like cross-coupling or
nonlinearity from the sensor's point of view. Identifying these terms,
and distinguishing genuine physics from noise artifacts, is the problem
this paper addresses.

Existing data-driven methods handle the full-state case well. SINDy
\cite{Brunton2016} recovers sparse governing equations by library
regression on state measurements and their time derivatives for the observed states; Weak
SINDy \cite{Messenger2021} extends this to noisy data by projecting
onto smooth test functions; physics-informed neural networks
\cite{Raissi2019} embed governing equations into the training loss and
have been applied across a broad range of forward and inverse problems.
What all these approaches share is the assumption that every state
entering the library is directly measurable. In rotordynamics, that
assumption fails by construction: the tilt angles $\alpha,\beta$
evolve dynamically and are coupled to the lateral displacements $x,y$
through the bearing geometry, but the sensors sit at the disk and see
only $x(t)$ and $y(t)$.

A growing body of work tackles exactly this partially observed
setting, and it is useful to place SPIRAL-PO against it. The classical
route is \emph{time-delay embedding}: Takens' theorem guarantees that
stacking delayed copies of a limited measurement reconstructs an
attractor diffeomorphic to that of the full-state system, an idea made
operational for linear-operator analysis by the Hankel alternative
view of Koopman (HAVOK) \cite{Brunton2017Havok}. Building on delay
coordinates, deep delay autoencoders learn a closed-form model from a
single measured channel \cite{Bakarji2023}, while related
autoencoder--SINDy architectures jointly discover latent coordinates
and the equations that govern them \cite{Champion2019}. A distinct
line absorbs the influence of the unobserved states by regressing
higher-order time derivatives of the measured variables
\cite{Somacal2022}. Orthogonally, classical state-estimation methods
--- extended and unscented Kalman filtering, moving-horizon estimation
--- reconstruct hidden states from a known model rather than
discovering the model itself. What the discovery methods share is that
they recover a model in surrogate coordinates --- delay,
latent, or derivative coordinates --- whose physical meaning must then
be inferred after the fact. SPIRAL-PO instead keeps the
observable coordinates fixed and projects the hidden-state dynamics
onto them as physically labeled library terms, so that, when the
hidden-state influence is adequately representable in those coordinates,
every discovered term retains a direct mechanical interpretation.

The present approach belongs to a long engineering tradition of
\emph{semi-empirical} modeling, especially prominent in unsteady
aerodynamics, where a model structure derived from physical insight and
dimensional reasoning is completed by coefficients fitted to data;
classical dynamic-stall models are the standard examples
\cite{LeishmanBeddoes1989,Tran1980ONERA,GomanKhrabrov1994}. The SPIRAL
method \cite{Ayoubi2027} was developed in that setting, for
nonlinear unsteady aerodynamics, and showed that iterative residual
augmentation from a physics seed can produce interpretable models. What
a statistically gated method adds to the semi-empirical tradition is a
reproducible admission rule and a coefficient uncertainty in place of
expert term selection.

This paper extends that idea to partially observed mechanical systems.
The contributions are as follows. First, we formulate the problem as the identification of an
\emph{observable-space closure}, and we state the class of systems for
which such a closure exists; for systems outside that class, where the
hidden-state influence depends on history or on latent state,
delay-coordinate, latent-state or observer-based methods are required
instead. Second, we build a mechanism-labeled candidate library from
both lateral displacement channels of a rotor, so that the signatures of
gyroscopic coupling and bearing nonlinearity transmitted by the hidden
tilt dynamics are representable. Third, we select terms with
tests designed for a numerically differentiated, serially correlated
signal: a weak (integral) formulation in which the measurement-noise part
of the errors is independent across test functions, effective-sample-size
corrections, block-bootstrap uncertainty and validation on independent
records; the false-admission rates of these tests are measured
empirically with knock-out controls. Fourth, we evaluate the method
against controls in which a mechanism is removed from the data, against
structural model mismatch, against ablations of every component and
against a delay-coordinate baseline.

The evaluations show that a residual neural network and a quasi-random
(Halton) projection stage do not
change which terms are selected once the statistics are valid, and that
gate statistics which ignore the serial correlation of the data admit
mechanisms that are absent from it. The library captures
only part of the hidden-state influence even without noise, so the
closure is an approximation whose size we quantify. The extended Kalman
filter that reconstructs the tilt angles relies on an assumed model of
the hidden dynamics and is therefore presented as a separate,
model-based step, not as an output of the identification.

The framework is applied to a four-degree-of-freedom (4-DOF) vertical
flexible rotor with cubic bearing stiffness \cite{Friswell2010}, in
which the tilt angles $\alpha(t)$, $\beta(t)$ are hidden and only the
lateral displacements $x(t)$, $y(t)$ at the disk are measured.
Section~\ref{sec:methodology} defines the system class and the
algorithm; Section~\ref{sec:identifiability} gives the identifiability,
non-identifiability and detection results for the idealized regression;
Section~\ref{sec:application} applies the method to the simulated rotor,
including controls, mismatch cases, ablations, a delay-coordinate
baseline and the model-based tilt reconstruction;
Section~\ref{sec:testrig} reports identification on the laboratory
rig; Sections~\ref{sec:discussion} and~\ref{sec:conclusions}
discuss related methods, limitations and conclusions.

\section{Methodology}
\label{sec:methodology}

\subsection{General problem formulation}
\label{sec:problem}

Before presenting the theory, we first define the term ``partially observed''.
Rather than absorbing the hidden states into a latent-space
representation or eliminating them by differentiation, we retain the
full-state equations and partition them, so that the role of each
unmeasured variable is structurally explicit from the outset.

In general, a nonlinear system can be represented as
\begin{equation}
  \dot{\bm{z}} = \bm{F}\bigl(\bm{z},\,\bm{u}(t),\,\bm{p}\bigr),
  \qquad
  \bm{z}(t) = \begin{Bmatrix} \bm{z}_{\mathrm{o}}(t) \\
                              \bm{z}_{\mathrm{h}}(t) \end{Bmatrix},
  \label{eq:general_ss}
\end{equation}
in which only the sub-state $\bm{z}_{\mathrm{o}}(t)$ is measured,
$\bm{z}_{\mathrm{h}}(t)$ is hidden, $\bm{u}(t)$ is a known input, and
$\bm{p}$ collects physical parameters.
Splitting the measured rows of $\bm{F}$ into a part expressible in
measured quantities alone and a remainder,
\begin{equation}
  \dot{\bm{z}}_{\mathrm{o}}
  = \bm{F}_{\mathrm{o}}^{\mathrm{obs}}
    \bigl(\bm{z}_{\mathrm{o}},\,\bm{u},\,\bm{p}\bigr)
  + \Delta\bigl(\bm{z}_{\mathrm{o}},\,\bm{z}_{\mathrm{h}},\,
                \bm{u},\,\bm{p}\bigr),
  \label{eq:general_split}
\end{equation}
collects every effect of the hidden sub-state on the measured
channels into a single coupling term $\Delta$.
This pair --- a measured derivative on the left, a structured but
unknown hidden-coupling term on the right --- is all that SPIRAL-PO
requires: the residual loop, the physics-structured library, and the
identifiability theory of Sect.~\ref{sec:identifiability} operate on
the form \eqref{eq:general_split}, regardless of the order of the
underlying system.
Mechanical and structural systems, the focus of this paper, 
arrive naturally in second-order form, so the remainder of the
development is phrased in generalized coordinates and their
derivatives.

Concretely, consider a second-order system with $n_q$ generalized
coordinates whose configuration vector partitions into an observable
part and a hidden part:
\begin{equation}
\begin{gathered}
  \bq(t) = \begin{Bmatrix} \bqo(t) \\ \bqh(t) \end{Bmatrix},\\[2pt]
  \bqo \in \R^{n_o},\quad
  \bqh \in \R^{n_h},\quad
  n_o + n_h = n_q.
\end{gathered}
  \label{eq:partition}
\end{equation}

We consider smooth mechanical systems of the form
\begin{equation}
  \bm{M}\,\ddot{\bq} = \bm{f}(t,\bm{p}) - \bm{h}(\bq,\dot{\bq},\bm{p}),
  \label{eq:full_eom_general}
\end{equation}
where $\bm{M}$ is a constant mass matrix, $\bm{f}$ a known forcing,
$\bm{h}$ a smooth (at least continuously differentiable) internal force,
and $\bm{p}$ a vector of physical parameters, including any known
scheduling parameter such as a spin speed. Nonsmooth forces (dry
friction, impacts, backlash), history-dependent forces (hysteresis,
internal memory) and delays are outside the scope of the development; a
hysteretic case is used in Sect.~\ref{ssec:mismatch} only to show how
the method behaves when this assumption is violated. Where convenient,
and for the rotor of Sect.~\ref{sec:application}, $\bm{h}$ is written in
the structured form
\begin{equation}
  \bm{h}(\bq,\dot{\bq},\bm{p})
  = \bm{D}(\bq,\dot{\bq},\bm{p})\,\dot{\bq}
  + \bm{K}(\bq,\bm{p})\,\bq,
  \label{eq:full_eom}
\end{equation}
with damping/gyroscopic and stiffness operators $\bm{D}$, $\bm{K}$.
This split is a modeling convenience: a
force such as $\sin q$ can be written as a state-dependent stiffness
only non-uniquely, and the block structure below is used only where the
split is physically natural.
Conformally with the partition~\eqref{eq:partition}, each operator
splits into observable ($o$) and hidden ($h$) blocks,
\begin{equation}
\begin{aligned}
  \bm{M} &=
  \begin{bmatrix} \bm{M}_{oo} & \bm{M}_{oh} \\
                  \bm{M}_{ho} & \bm{M}_{hh} \end{bmatrix},
  \quad
  \bm{D} =
  \begin{bmatrix} \bm{D}_{oo} & \bm{D}_{oh} \\
                  \bm{D}_{ho} & \bm{D}_{hh} \end{bmatrix},\\[3pt]
  \bm{K} &=
  \begin{bmatrix} \bm{K}_{oo} & \bm{K}_{oh} \\
                  \bm{K}_{ho} & \bm{K}_{hh} \end{bmatrix},
\end{aligned}
  \label{eq:blocks}
\end{equation}
where the $oo$ blocks are $n_o \times n_o$, the $oh$ blocks
$n_o \times n_h$, the $ho$ blocks $n_h \times n_o$, and the $hh$
blocks $n_h \times n_h$. The observable--observable block $\bm{M}_{oo}$
maps observable accelerations to observable forces, while the
observable--hidden blocks $\bm{D}_{oh}$ and $\bm{K}_{oh}$ are precisely
the channels through which the unmeasured coordinates $\bqh$ act on the
measured ones. Projecting~\eqref{eq:full_eom} onto its observable rows
and assuming block-diagonal inertia ($\bm{M}_{oh}=\bm{0}$), isolates $\bm{M}_{oo}\ddbqo$ and collects
 hidden-block actions $\bm{D}_{oh}\dbqh$, $\bm{K}_{oh}\bqh$ together
with nonlinear cross-coupling in the term $\Delta$
of~\eqref{eq:obs_eom}.
Only $\bqo(t)$ (and its time derivatives, estimated from data) is
available.

In the structured form \eqref{eq:full_eom}, $\bm{D}$ also carries any
velocity-proportional coupling, including speed-dependent gyroscopic
effects, so that a scheduling parameter enters through $\bm{p}$. The
concrete rotordynamic instance, with an explicit gyroscopic matrix and
cubic bearing stiffness, is given in Sect.~\ref{sec:application}.

\subsubsection{Observable equations of motion}

The partition above is useful only if it leads to equations that
can be written solely in terms of measurable quantities.
Projecting the full equations onto the observable rows achieves
exactly this: it isolates the dynamics of $\bqo$ while collecting
all hidden-state influence into a single, structured residual term
$\Delta$.
That term is unknown and inherits the
physical constraints of the hidden dynamics, a property that
SPIRAL-PO will exploit in the discovery step.

To separate linear structure from nonlinear coupling, decompose the
operators into constant and state-dependent parts,
\begin{equation}
\begin{aligned}
  \bm{D}(\bq,\dot{\bq},\bm{p}) &= \bar{\bm{D}}(\bm{p})
    + \bm{D}_{\mathrm{nl}}(\bq,\dot{\bq},\bm{p}),\\
  \bm{K}(\bq,\bm{p}) &= \bar{\bm{K}}(\bm{p})
    + \bm{K}_{\mathrm{nl}}(\bq,\bm{p}),
\end{aligned}
  \label{eq:DK_split}
\end{equation}
and partition both parts into blocks conformal with \eqref{eq:partition},
\begin{equation}
\begin{aligned}
  \bar{\bm{D}} &=
  \begin{bmatrix} \bm{D}_{oo} & \bm{D}_{oh} \\
                  \bm{D}_{ho} & \bm{D}_{hh} \end{bmatrix},\\[3pt]
  \bm{D}_{\mathrm{nl}} &=
  \begin{bmatrix} \bm{D}_{\mathrm{nl},oo} & \bm{D}_{\mathrm{nl},oh} \\
                  \bm{D}_{\mathrm{nl},ho} & \bm{D}_{\mathrm{nl},hh}
  \end{bmatrix},
\end{aligned}
  \label{eq:block_partition}
\end{equation}
and analogously for $\bar{\bm{K}}$, $\bm{K}_{\mathrm{nl}}$, where the
nonlinear blocks retain their dependence on the full state
$(\bq,\dot{\bq})$. Taking the observable block row of
\eqref{eq:full_eom} (with $\bm{M}_{oh}=\bm{0}$) gives

\begin{equation}
\begin{aligned}
  \bm{M}_{oo}\ddbqo
  &= \bm{f}_o - \bm{D}_{oo}\dbqo - \bm{K}_{oo}\bqo\\
  &\quad - \bm{D}_{oh}\dbqh - \bm{K}_{oh}\bqh\\
  &\quad - g(\bqo,\bqh,\dbqo,\dbqh),
\end{aligned}
  \label{eq:obs_row}
\end{equation}
where
\begin{equation}
\begin{aligned}
  g(\bqo,\bqh,\dbqo,\dbqh)
  &= \bm{D}_{\mathrm{nl},oo}\,\dbqo
  + \bm{D}_{\mathrm{nl},oh}\,\dbqh\\
  & + \bm{K}_{\mathrm{nl},oo}\,\bqo
  + \bm{K}_{\mathrm{nl},oh}\,\bqh,
\end{aligned}
  \label{eq:g_def}
\end{equation}
collects all nonlinear cross-coupling: even the terms acting on
$\bqo$, $\dot{\bqo}$ involve the hidden states through the
state dependence of the nonlinear blocks. Let's write \eqref{eq:obs_row} in the following form
\begin{equation}
\begin{aligned}
\bm{M}_{oo}\,\ddbqo
={}& \bm{f}_o(t,\bm{p})
-\bm{D}_{oo}\,\dbqo
-\bm{K}_{oo}\,\bqo \\
&-\bm{\Delta}(\bqo,\bqh,\dbqo,\dbqh),
\end{aligned}
\label{eq:obs_eom}
\end{equation}

\begin{equation}
\begin{aligned}
\bm{\Delta}(\bqo,\bqh,\dbqo,\dbqh)
&\coloneqq{}
\bm{D}_{oh}\,\dbqh
+\bm{K}_{oh}\,\bqh \\
&+\bm{g}(\bqo,\bqh,\dbqo,\dbqh),
\end{aligned}
\label{eq:hidden_coupling}
\end{equation}
where $\bm{\Delta}$ is unknown and depends on the hidden states. Along a
solution of \eqref{eq:full_eom_general}, $\bqh(t)$ and $\dbqh(t)$ are
determined by the initial conditions and the past forcing, not in
general by the instantaneous observed state: two trajectories passing
through the same $(\bqo,\dbqo,\bm{p})$ can carry different hidden
states, and then $\bm{\Delta}$ evaluated along them differs. The
hidden-state contribution is therefore, in general, history dependent,
and the identification problem considered in this paper is meaningful
only for systems in which this dependence can be approximated
instantaneously.

\paragraph{System class.} We say that a partially observed system admits
an \emph{instantaneous observable closure} on a record if there exist a
finite set of functions $\{\phi_\ell\}$ of the observable regressor
$\bxi=(\bqo,\dbqo,\bm{p}_{\mathrm{sched}})$, coefficients $\btheta^*$
and a relative remainder $\epsilon_0<1$ such that, along the realized
trajectory,
\begin{equation}
  \bm{\Delta}(t) = \sum_{\ell}\theta^*_\ell\,\phi_\ell\bigl(\bxi(t)\bigr)
  + \eta(t),
  \qquad
  \frac{\norm{\eta}_{L^2}}{\norm{\bm{\Delta}}_{L^2}}\le\epsilon_0 .
  \label{eq:closure_class}
\end{equation}

This is a structural hypothesis about the system. It fails, for example, when the hidden dynamics are slow relative to the
observable motion or carry memory (hysteresis), and in such cases a
delay-coordinate, latent-state or observer-based formulation, or
additional sensing, is required. All theoretical results of
Sect.~\ref{sec:identifiability} are stated for this class only. For the
rotor of Sect.~\ref{sec:application} we estimate $\epsilon_0$ directly
from noise-free simulations (Sect.~\ref{ssec:capacity}); it is far from
zero, so the closure identified there is an approximation, and its
quality is quantified in Sect.~\ref{ssec:capacity}.

\subsubsection{Identification objective}

With the problem structure in place, the identification objective can
be stated concisely. The goal is not to reconstruct $\bqh$ directly, instead, to determine which terms belong in $\Delta$ and
what their coefficients are, using only the observable trajectory.

Measurements of the observable coordinates are corrupted by noise,
\begin{equation}
  \tbqo(t_k) = \bqo(t_k) + \varepsilon_k,
  \qquad k = 1,\ldots,N,
  \label{eq:measurement}
\end{equation}
where $\varepsilon_k \sim \mathcal{N}(\bm{0},\sigma^2\bm{I})$ is
i.i.d.\ measurement noise on the displacements. The regression target,
however, is a numerically differentiated acceleration: differentiation
filters $\varepsilon_k$, so the errors of the acceleration equation are
serially correlated (for a Savitzky--Golay derivative of window $w$
they form a moving average of order $w-1$). The estimation and test
procedures of Sect.~\ref{sec:method} are built for this error
structure.
Given \eqref{eq:measurement} and a known excitation
$\bm{f}(t,\bm{p})$, the identification objective is the pair
$(\fphys,\btheta^*)$: a functional form
$\fphys(\bqo,\dbqo,t;\btheta)$ and coefficients
\begin{equation}
\begin{aligned}
  \btheta^*
  &= \arg\min_{\btheta}\;
    \sum_{k=1}^{N}
    \bigl\|\bm{r}(t_k;\btheta)\bigr\|_2^2,\\
  \bm{r}(t_k;\btheta)
  &= \ddtbqo(t_k)\\
  &\quad - \fphys\Bigl(\tbqo(t_k),\dtbqo(t_k),t_k;\btheta\Bigr),
\end{aligned}
  \label{eq:ident_obj}
\end{equation}
where the residual $\bm{r}$ is the portion of the measured
acceleration unexplained by the physics model. The estimate is reported
with uncertainties that account for the correlated errors
(Sect.~\ref{ssec:uncertainty}).

\subsection{The SPIRAL-PO algorithm}
\label{sec:method}

The method rests on one organizing question: which physically labeled
terms must be added to a minimal physics model to explain the measured
dynamics, and how sure can we be of each one? A term is admitted only if
its coefficient is significant under a test whose error model matches
the data, and only if the augmented model predicts data it was not
fitted to. Every admitted term carries a coefficient, an uncertainty and
a physical label, so the final model can be checked in the same way as
a conventional engineering model. A residual neural network serves only as a diagnostic of what the
library leaves unexplained (Step~7), and the inference is carried out on
a weak (integral) form of the equations; Sect.~\ref{ssec:gated}
describes a gated network formulation in which the network enters the
selection step, and Sect.~\ref{ssec:ablation} compares the two.

\subsubsection{Observable regressor vector}
\label{ssec:regressor}

The observable regressor is
\begin{equation}
  \bxi(t)
  = \bigl(\bqo^{\mathsf{T}},\;\dbqo^{\mathsf{T}},\;
    \bm{p}_{\mathrm{sched}}^{\mathsf{T}}\bigr)^{\mathsf{T}}
  \;\in\;\R^{n_\xi},
  \label{eq:regressor}
\end{equation}
where $\bm{p}_{\mathrm{sched}}$ contains known scheduling parameters
(for example the spin speed $\OO$ of a rotor). The hidden coordinates
$\bqh$ do not appear in $\bxi$; velocities $\dbqo$ are estimated from
$\tbqo$ as described in Sect.~\ref{sec:deriv}.

\subsubsection{Step 1 --- Physics seed}
\label{ssec:seed}

SPIRAL-PO begins from the minimal model expressible in the observable
coordinates alone,
\begin{equation}
\begin{split}
  \fphys^{(0)}(\bxi;\,\btheta^{(0)})
  &= \bm{M}_{oo}^{-1}
    \Bigl[
      \bm{f}_o(t)\\
  &\quad - \bm{D}_{oo}^{(0)}\,\dbqo
      - \bm{K}_{oo}^{(0)}\,\bqo
    \Bigr],
\end{split}
  \label{eq:seed}
\end{equation}
with the forcing, damping and stiffness coefficients estimated from the
data. All coupling to hidden states, all nonlinear terms and all
scheduling-dependent terms start inactive and must be earned.

\subsubsection{Step 2 --- Mechanism-labeled candidate library}
\label{ssec:library}

Candidate terms are not drawn from a generic polynomial expansion but
derived from the mechanisms that can act through the hidden-coupling
term $\bm{\Delta}$ of \eqref{eq:hidden_coupling}: each candidate is
obtained by projecting a specific hidden-state mechanism (for the rotor:
bearing coupling, gyroscopic transmission, cubic bearing stiffness) onto
the observable coordinates. The library at iteration $n$ is
\begin{equation}
\calL^{(n)}
=
\calL_{\mathrm{known}}
\cup
\calL_{\mathrm{hidden}}^{(n)}
\setminus
\mathcal{A}^{(n)},
\label{eq:library}
\end{equation}
where $\calL_{\mathrm{known}}$ contains terms of the observable
dynamics, $\calL_{\mathrm{hidden}}^{(n)}$ the projected hidden-state
candidates and $\mathcal{A}^{(n)}$ the terms already in the model.

Physical reasoning contributes three things to this construction.
Dimensional consistency applies to the product of a coefficient and a
basis function, so it constrains the units of each coefficient rather
than the choice of basis functions (the rotor library contains $x$,
$\OO\dot y$ and $(x^2+y^2)x$, whose units differ). The derivation
provides (i) a mechanism label for every candidate, (ii) the expected units, sign and order of
magnitude of its coefficient, which give a plausibility check after
identification, and (iii) a principled reason to omit mechanisms that
cannot arise from the assumed physics. The last point is a modeling
choice, and Sect.~\ref{ssec:mismatch} tests its consequences when the
data contain a mechanism the library lacks.

Candidates that are linear combinations of the seed regressors along the
record cannot be estimated separately and are removed before
regression. For the rotor this matters: under a constant-deceleration
coast-down, $\dot\OO$ is constant and $\dot\OO\,x$ is a multiple of the
seed column $x$.

\subsubsection{Step 3 --- Weak-form regression}
\label{ssec:weak}

Regressing a numerically differentiated acceleration onto the library
has two drawbacks: second differentiation amplifies the measurement
noise, and it makes the regression errors serially correlated, so that
classical standard errors and $t$-statistics are too small. Both are
removed by a weak (integral) formulation \cite{Messenger2021}. Let
$\psi_m(t)$, $m=1,\ldots,M$, be smooth test functions with compact,
mutually disjoint supports $[c_m-h,\,c_m+h]$ (here
$\psi_m(t)=(1-\tau^2)^4$, $\tau=(t-c_m)/h$). Multiplying the
$i$-th observable equation by $\psi_m$ and integrating by parts moves
every time derivative of the measured signal onto the test function,
\begin{equation}
\begin{aligned}
  \int\psi_m\,\ddot q_{o,i}\,dt &= \int\ddot\psi_m\,q_{o,i}\,dt,\\
  \int\psi_m\,\dot q_{o,i}\,dt &= -\int\dot\psi_m\,q_{o,i}\,dt ,
\end{aligned}
  \label{eq:weak_ibp}
\end{equation}
and analogously for library terms that are total derivatives or linear
in a velocity, e.g.\ $\int\psi_m\OO\dot y\,dt=-\int(\dot\psi_m\OO+
\psi_m\dot\OO)\,y\,dt$. The result is a linear system
\begin{equation}
\begin{split}
  \bm{g} &= \bm{G}\,\btheta + \bm{e},\\
  g_m &= \int\ddot\psi_m\,\tilde q_{o,i}\,dt,\\
  G_{m\ell} &= \int\psi_m\,\phi_\ell\,dt ,
\end{split}
  \label{eq:weak_system}
\end{equation}
whose columns contain the seed regressors and the candidates. Because
the supports are disjoint and the displacement noise is white, the
errors $e_m$ of different test functions are independent, and each is a
weighted sum of many noise samples, hence close to Gaussian. This
independence concerns the measurement-noise part of $e_m$. The library
remainder $\eta$ of \eqref{eq:closure_class} is a smooth function of
time and is correlated across neighboring test functions; when the
noise is small and $\eta$ dominates the errors, the standard errors of
\eqref{eq:weak_system} are therefore optimistic
(Sect.~\ref{ssec:selection}). The system
\eqref{eq:weak_system} is therefore much closer to the idealized
regression of Sect.~\ref{sec:identifiability} than the pointwise
acceleration regression. Velocities that cannot be integrated by parts
(e.g.\ in $(x^2+y^2)\dot x$) enter through the smoothed derivative of
Sect.~\ref{sec:deriv}.

\subsubsection{Step 4 --- Gate~1: Bonferroni selection}
\label{ssec:gate1}

Each candidate $\ell$ is tested for $H_0:\theta_\ell=0$ against
$H_1:\theta_\ell\neq0$ with the statistic
$\hat t_\ell=\hat\theta_\ell/\mathrm{SE}(\hat\theta_\ell)$ from the OLS
solution of \eqref{eq:weak_system}, all candidates being fitted
jointly with the seed. Because $K$ candidates are tested together, the
family-wise error rate is held at $\alpha=0.05$ by the Bonferroni
correction \cite{Dunn1961,Miller1981}: a candidate passes Gate~1 if
\begin{equation}
  |\hat t_\ell| > z_{1-\alpha/(2K)} ,
  \label{eq:gate1}
\end{equation}
with $z_\gamma$ the standard normal quantile. With about 160 test
functions, the corresponding Student quantile is larger by about 1.5\%
($2.78$ instead of $2.73$ for $K=8$). The family-wise error rate is
controlled separately in each observable equation. Proposition~\ref{prop:validity} gives the conditions
under which this screen detects true terms and controls false
admissions.

\subsubsection{Step 5 --- Gate~2: validation on held-out data}
\label{ssec:gate2}

Significance on the fitting record is necessary but not sufficient. The
model formed by the seed and the Gate-1 terms is refitted and must
reduce the post-seed residual on data not used for fitting: an
independent record when one is available (a second run, or for
simulations a second noise realization with a different unbalance
phase), otherwise contiguous time blocks held out in turn, with a gap of
one differentiation window on each side to prevent leakage through the
correlated errors. The decision uses the pooled held-out error over all
blocks rather than a majority vote, because an amplitude-dependent term
acts only in the high-amplitude part of a coast-down and would be
outvoted by blocks in which it is inactive. If the held-out error does
not decrease, the selected terms are rejected.

\subsubsection{Step 6 --- Coefficient uncertainty}
\label{ssec:uncertainty}

Two uncertainty estimates are reported for every coefficient: the
standard error of the weak-form regression, valid under the independence
argument of Step~3, and a moving-block bootstrap \cite{Kunsch1989} of
the pointwise regression on the measured trajectory, in which
contiguous blocks of rows (regressors and target) are resampled with
replacement. The block length is four differentiation windows, longer
than the correlation length of the differentiated noise, so the
bootstrap preserves the serial correlation that classical standard
errors ignore.

\subsubsection{Step 7 --- Residual predictability check}
\label{ssec:guard}

\begin{figure}[!t]
\centering
\newcommand{\tanhglyph}[2]{%
  \begin{scope}[shift={(#1,#2)}]
    \clip (0,0) circle (0.30);
    \draw[gray!55,very thin] (-0.26,0) -- (0.26,0);
    \draw[gray!55,very thin] (0,-0.26) -- (0,0.26);
    \draw[black!75,line width=0.9pt]
      plot[domain=-1:1,samples=31,smooth] (\x*0.22,{0.18*tanh(3*\x)});
  \end{scope}}
\newcommand{\linglyph}[2]{%
  \begin{scope}[shift={(#1,#2)}]
    \clip (0,0) circle (0.30);
    \draw[gray!55,very thin] (-0.26,0) -- (0.26,0);
    \draw[black!75,line width=0.9pt] (-0.24,-0.18) -- (0.24,0.18);
  \end{scope}}
\resizebox{0.95\columnwidth}{!}{%
\begin{tikzpicture}[
  inn/.style={circle,draw=blue!60!black,fill=blue!12,minimum size=0.62cm,inner sep=0pt,font=\footnotesize},
  hid/.style={circle,draw=orange!75!black,fill=orange!18,minimum size=0.62cm,inner sep=0pt},
  outn/.style={circle,draw=green!55!black,fill=green!16,minimum size=0.62cm,inner sep=0pt},
  conn/.style={gray!45,line width=0.2pt},
  hdr/.style={font=\small\bfseries}]
  \def\xI{0} \def\xHa{3} \def\xHb{6} \def\xO{9}
  \foreach \a in {2,1,0,-1,-2}{\foreach \b in {2,1,0,-1,-2}{\draw[conn] (\xI,\a)--(\xHa,\b);}}
  \foreach \a in {2,1,0,-1,-2}{\foreach \b in {2,1,0,-1,-2}{\draw[conn] (\xHa,\a)--(\xHb,\b);}}
  \foreach \a in {2,1,0,-1,-2}{\draw[conn] (\xHb,\a)--(\xO,0);}
  \node[inn] at (\xI,2)  {$x$};
  \node[inn] at (\xI,1)  {$y$};
  \node[inn] at (\xI,0)  {$\dot x$};
  \node[inn] at (\xI,-1) {$\dot y$};
  \node[inn] at (\xI,-2) {$\OO$};
  \foreach \y in {2,1,0,-1,-2}{\node[hid] at (\xHa,\y){}; \tanhglyph{\xHa}{\y}}
  \foreach \y in {2,1,0,-1,-2}{\node[hid] at (\xHb,\y){}; \tanhglyph{\xHb}{\y}}
  \node[outn] (O) at (\xO,0){}; \linglyph{\xO}{0}
  \node[right=5pt of O,font=\footnotesize] {$\hat r_i$};
  \node[hdr] at (\xI,2.85)  {Input};
  \node[font=\footnotesize] at (\xI,-2.8) {$\bxi$};
  \node[hdr] at (\xHa,2.85) {Hidden 1};
  \node[hdr] at (\xHb,2.85) {Hidden 2};
  \node[font=\footnotesize] at ({(\xHa+\xHb)/2},-2.8) {$\tanh$, 12 neurons per layer};
  \node[hdr] at (\xO,2.85)  {Output};
  \node[font=\footnotesize] at (\xO,-2.8) {linear};
\end{tikzpicture}}
\caption{Residual network used in the predictability check (Step~7)
for the 4-DOF rotor. The observable quantities
$\bxi=(x,y,\dot x,\dot y,\OO)^\mathsf{T}$ enter the input layer
(blue); two hidden layers use $\tanh$ activations (orange, with the
activation curve drawn inside each neuron; five of the twelve neurons
per layer are drawn); the linear output (green) is the predicted
residual $\hat r_i$ of observable equation $i$. One network is fitted
per equation.}
\label{fig:nn_arch}
\end{figure}

After Steps~3--6, the remaining residual is examined for structure the
library has not captured. A small residual network (two $\tanh$ layers,
Fig.~\ref{fig:nn_arch}, trained with Adam \cite{Kingma2015} and early stopping) is fitted to
the residual as a function of $\bxi$ on alternating contiguous chunks of
the record and scored on the held-out chunks; a one-sided $t$-test of
the per-chunk log ratio of mean-square errors decides whether the
network predicts the residual better than zero. A predictable residual
indicates that the closure is incomplete ($\epsilon_0>0$ in
\eqref{eq:closure_class}); the library may then be extended from a
reserve set and Steps~3--7 repeated, which constitutes the
residual-augmentation loop. Autocorrelation-based whiteness tests such as
Ljung--Box \cite{LjungBox1978} are poorly suited to this role here: the
differentiated noise is itself correlated, and when the residual is
sampled one differentiation window apart, so that the noise is
decorrelated, the remaining closure signal (a chirp at the spin
frequency) is nearly white as well. Ljung--Box statistics are therefore
reported only as diagnostics.

\subsubsection{Gated network formulation}
\label{ssec:gated}

In the gated network formulation, the network enters the selection
step, which changes Steps~3--5 and~7. It fits the residual network at every iteration, evaluates it on
$N_H=2\times10^4$ Halton points \cite{Halton1960} filling the bounding
box of $\bxi$, regresses the network output on the library there and
screens candidates with Student-$t$ tests on $N_H$ degrees of freedom;
it then admits terms by partial $F$-tests on the pointwise regression
with $N-p$ denominator degrees of freedom and an amplitude-stratified
$K$-fold majority vote, and stops when a Ljung--Box test on the
validation residual no longer rejects whiteness. Three properties limit
this formulation.
The network outputs on the Halton grid are deterministic, so their
regression residuals are approximation error, not sampling noise, and
the $t$-distribution does not apply; many Halton points lie away from the
visited trajectory, where the network is unconstrained; and the
pointwise tests treat roughly $N$ serially correlated samples as
independent, whereas the effective number of independent samples is
about $N/w$ for a derivative window $w$. Amplitude-stratified folds also
interleave neighboring samples in time. Both formulations are included in the
released code, and Sect.~\ref{ssec:ablation} compares them, together
with a variant in which the network is evaluated only in a tube around
the trajectory.

\subsection{Derivative estimation for observable states}
\label{sec:deriv}

Velocities, and for the pointwise regressions and the bootstrap also
accelerations, are estimated with Savitzky--Golay differentiators of
polynomial order four \cite{SavitzkyGolay1964}; the window length is
chosen from $\{21,31,51\}$ samples by predicting every tenth sample from
its neighbors. The differentiated noise is then a moving average of
order $w-1$, which is the correlation structure the tests of
Sect.~\ref{sec:method} account for. The weak formulation needs no
acceleration estimate at all. For a steady periodic response at a known
constant speed, a multi-harmonic Fourier fit,
\begin{equation}
  \tilde{q}_{o,i}(t) \approx \sum_{k} A_k\cos(k\OO t)
  + B_k\sin(k\OO t),
  \label{eq:fourier_fit}
\end{equation}
can be differentiated analytically instead.

\subsection{Diagnostic statistics}
\label{sec:diagnostics}

Three non-blocking statistics are reported for each fitted model,
computed with unit-RMS (standardized) columns so that terms of very
different natural magnitude do not degrade the numerics. The predicted
residual sum of squares \cite{montgomery2018},
\begin{equation}
  \mathrm{PRESS}
  = \sum_{i=1}^{N}
    \left(\frac{e_i}{1 - h_{ii}}\right)^{\!2},
  \quad
  h_{ii} = \norm{\bm{Q}(i,:)}_2^2,
  \label{eq:press}
\end{equation}
uses the economy QR decomposition $\bPhi=\bm{Q}\bm{R}$ of the
regressor matrix, with $e_i$ the $i$-th residual and $h_{ii}$ its
leverage; the variance inflation factor
\begin{equation}
  \mathrm{VIF}_j = \norm{[\bm{R}^{-1}]_{j,:}}_2^2
  \cdot (N-1)
  \label{eq:vif}
\end{equation}
flags collinearity; and the adjusted coefficient of determination
\begin{equation}
  R^{2}_{\mathrm{adj}}
  = 1 - \frac{N-1}{N-p}
    \cdot \frac{\mathrm{RSS}}{SS_T}
  \label{eq:radj}
\end{equation}
penalizes model size. The conditioning of the full signature matrix is
reported separately through its singular values (Sect.~\ref{ssec:ident_rotor}).

\subsection{Algorithm summary}
\label{sec:algorithm}

Algorithm~\ref{alg:spiral_po} states the procedure and
Fig.~\ref{fig:spiral_loop} its control flow.

\begin{figure}[!t]
\centering
\resizebox{0.95\columnwidth}{!}{%
\begin{tikzpicture}[
  >={Stealth[length=2.2mm]},
  node distance=0.55cm,
  box/.style={draw,rounded corners=2pt,fill=gray!6,align=center,font=\small,text width=6.2cm,inner sep=4pt,minimum height=0.7cm},
  io/.style={box,fill=blue!8},
  dec/.style={box,fill=orange!12},
  outbox/.style={box,fill=green!10},
  fl/.style={->,thick}]
  \node[io]  (data) {Measured $\tilde{\bm q}_{\mathrm{obs}}$, known $\OO(t)$; velocities (Sect.~\ref{sec:deriv})};
  \node[box] (seed) [below=of data] {Step 1: physics seed $\fphys^{(0)}$};
  \node[box] (lib)  [below=of seed] {Step 2: mechanism-labeled library $\calL$; remove terms collinear with the seed};
  \node[box] (weak) [below=of lib]  {Step 3: weak-form system $\bm g=\bm G\btheta+\bm e$ (seed $+$ all candidates)};
  \node[dec] (g1)   [below=of weak] {Step 4, Gate 1: Bonferroni $|\hat t_\ell|>z_{1-\alpha/(2K)}$};
  \node[dec] (g2)   [below=of g1]   {Step 5, Gate 2: held-out error decreases?};
  \node[box] (unc)  [below=of g2]   {Step 6: weak-form SE and block-bootstrap SE};
  \node[dec] (g3)   [below=of unc]  {Step 7: residual still predictable from $\bxi$?};
  \node[outbox] (fin)  [below=of g3]   {Annotated model with uncertainties and estimate of what remains};
  \foreach \a/\b in {data/seed,seed/lib,lib/weak,weak/g1,g1/g2,g2/unc,unc/g3,g3/fin}{\draw[fl] (\a)--(\b);}
  \draw[fl] (g3.east) -- ++(0.5,0) |- node[pos=0.25,right,font=\scriptsize,align=left]{yes, reserve\\terms left:\\extend $\calL$} (lib.east);
\end{tikzpicture}}
\caption{SPIRAL-PO procedure. Terms are selected on a weak
(integral) form of the observable equations, whose test-function errors
are independent, admitted only if they reduce the error on held-out
data, and reported with two uncertainty estimates; a residual network is
used only to judge whether structure remains that the library has not
captured.}
\label{fig:spiral_loop}
\end{figure}

\begin{algorithm}[!t]
\caption{SPIRAL-PO}
\label{alg:spiral_po}
\begin{algorithmic}[1]
\Require
  $\{\tbqo(t_k)\}_{k=1}^N$, known forcing $\bm{f}(t)$ and
  $\bm{p}_{\mathrm{sched}}(t)$, seed form $\fphys^{(0)}$, library
  $\calL$ with reserve set, level $\alpha$, test-function half-width $h$.
\Ensure
  Selected terms, coefficients, weak-form and bootstrap standard
  errors, held-out gain, residual predictability.
\State Estimate $\dbqo$ (and $\ddbqo$) by Savitzky--Golay
       differentiation (Sect.~\ref{sec:deriv}).
\Repeat
  \State Remove candidates collinear with the seed regressors.
  \State Build the weak system \eqref{eq:weak_system} for the seed and
         all candidates; solve by QR-OLS.
  \State Gate 1: keep candidates with
         $|\hat t_\ell|>z_{1-\alpha/(2K)}$.
  \State Refit seed $+$ kept terms; Gate 2: compare held-out post-seed
         residual with and without them; reject them if it does not
         decrease.
  \State Report weak-form and moving-block-bootstrap standard errors.
  \State Fit the residual network; test held-out predictability of the
         residual.
  \If{residual predictable and reserve terms remain}
     \State add reserve terms to $\calL$
  \EndIf
\Until{residual not predictable, or no reserve terms remain}
\State \Return model, uncertainties, diagnostics.
\end{algorithmic}
\end{algorithm}

\section{Identifiability theory}
\label{sec:identifiability}

\subsection{Motivation}

Before running SPIRAL-PO on any dataset it is essential to know
whether the hidden-coupling terms are recoverable in principle
from the observable projection residual.
Without such a guarantee, the method may converge to a misidentified
model even in the noise-free limit.
This section establishes sufficient conditions for identifiability
(Theorem~\ref{thm:ident}) and a matching non-identifiability result for
the excitation requirement (Proposition~\ref{thm:impossibility}). These
are standard least-squares facts specialized to the observable signature
matrix; their value here is that they make the role of the library
remainder $\epsilon_0$ and of the excitation $\mu_\Phi$ explicit.
Theorem~\ref{thm:ident} applies to systems that admit an instantaneous
observable closure \eqref{eq:closure_class}; it is not an
identifiability result for general partially observed nonlinear
systems.

\paragraph{Scope of identifiability.}
The results below establish identifiability of the \emph{coefficients
of the projected hidden-state coupling} within the prescribed
observable library, and only up to the persistent
library-approximation bias of Theorem~\ref{thm:ident}. They do
\emph{not} imply unique recovery of the hidden-state trajectory
$\bqh(t)$ itself, nor of the complete governing equations; recovering
$\bqh(t)$ requires the additional observability assumptions and the
state-estimation step of Sect.~\ref{ssec:ekf_reconstruction}.

It is useful to separate what is \emph{proved} from what is
\emph{empirically assessed}. The theory of this section concerns an
idealized regression with known regressors: the hidden coupling admits
the library representation \eqref{eq:closure_class}, the excitation
satisfies a persistency condition, and the regression errors are
independent with common variance (Gaussian for the test statements).
The weak-form system \eqref{eq:weak_system} approximately satisfies the
error assumptions, because its test-function errors are independent;
the pointwise regression on differentiated data does not, and the
theorems do not cover numerical differentiation of the regressors,
errors in the regressors themselves, or data-dependent selection. These
elements are assessed numerically (Sect.~\ref{ssec:hidden_state_validation}),
and the theorems describe the estimation core of the procedure.

\subsection{Definitions}

Several objects recur throughout the analysis. For a candidate term
$\phi_\ell(\bxi) \in \calL_{\mathrm{hidden}}$, its \emph{observable
signature} is the time history
\begin{equation}
  s_\ell(t) \;\triangleq\; \phi_\ell\!\bigl(\bxi(t)\bigr).
  \label{eq:signature}
\end{equation}

Stacking these signatures for $p$ candidate terms across $N$ time
samples gives the \emph{signature matrix}
\begin{equation}
\begin{aligned}
      \bPhi_{\bxi}
  = \bigl[&\phi_1(\bxi(t_1)),\ldots,\phi_p(\bxi(t_1));\;
    \ldots;\;\\
    &\phi_1(\bxi(t_N)),\ldots,\phi_p(\bxi(t_N))\bigr],
\end{aligned}
  \label{eq:sig_matrix}
\end{equation}
so that $\Delta \approx \bPhi_{\bxi}\,\btheta^*$. We say that the record
is \emph{persistently exciting for the library} if the normalized Gram
matrix of the signatures is uniformly positive definite,
\begin{equation}
  \lambda_{\min}\!\left(N^{-1}\bPhi_{\bxi}^{\mathsf{T}}\bPhi_{\bxi}\right)
  \;\ge\; \mu_\Phi > 0 ,
  \label{eq:PE}
\end{equation}
with columns scaled to unit RMS. The $(k,\ell)$ entry of this Gram
matrix is the time-averaged product of two candidate signatures, so
\eqref{eq:PE} states that no signature is reproducible as a combination
of the others along the trajectory, which is what allows their
coefficients to be separated. This condition is stated at the level of
the library because excitation of the raw regressor $\bxi$ (a
positive definite covariance of $x,y,\dot x,\dot y,\OO$) does not imply
linear independence of nonlinear functions such as $x$, $x^3$,
$(x^2+y^2)x$, $\OO^2x$ and $\OO\dot y$, which is exactly why some rotor
candidates are nearly collinear on a coast-down.

\subsection{Identifiability: when is recovery possible?}

When can the hidden-coupling coefficients $\btheta^*$ be recovered
uniquely from the observable residual? Two conditions suffice: the
hidden influence must be expressible in the library up to a small
remainder (the closure class \eqref{eq:closure_class}), and the
trajectory must excite all candidate directions \eqref{eq:PE}. That rich excitation is also
\emph{necessary} is established by Proposition~\ref{thm:impossibility}. The
proof is the standard least-squares consistency and bias--variance
argument \cite{LehmannCasella1998}; what is specific to SPIRAL-PO is that
the excitation quality $\mu_\Phi$ governs both error terms.

\begin{theorem}[Coefficient consistency and error bound]
\label{thm:ident}
Let the full system satisfy \eqref{eq:full_eom_general}, let the
regressors be known exactly, and let the regression errors be
independent with zero mean and variance $\sigma_\varepsilon^2$. Suppose:
\begin{enumerate}[label=\textup{(C\arabic*)}]
  \item Observable closure. The hidden coupling evaluated along the
        record satisfies \eqref{eq:closure_class}, i.e.\
        $\Delta(t)
        = \sum_{\ell=1}^{p} \theta_\ell^*\,\phi_\ell(\bxi(t)) + \eta(t)$,
        with $\norm{\eta}_{L^2}/\norm{\Delta}_{L^2} \le \epsilon_0$. Here $\norm{\cdot}_{L^2}$ denotes the $L^2(0,T)$ signal norm,
$\norm{v}_{L^2} = \bigl(\int_0^T |v(t)|^2\,dt\bigr)^{1/2}$, so that
$\norm{\eta}_{L^2}/\norm{\Delta}_{L^2}$ is the relative energy of the
unmodeled remainder $\eta$ against the total hidden coupling $\Delta$.

  \item Persistent excitation.
        $\lambda_{\min}(N^{-1}\bPhi_{\bxi}^{\mathsf{T}}\bPhi_{\bxi})
        \ge \mu_\Phi > 0$.
\end{enumerate}

Then the OLS estimator
$\hat{\btheta}=(\bPhi_{\bxi}^{\mathsf{T}}\bPhi_{\bxi})^{-1}
\bPhi_{\bxi}^{\mathsf{T}}\ddtbqo$ is uniquely defined and satisfies
\begin{equation}
  \norm{\hat{\btheta}-\btheta^*}_2
  \;\le\;
  \underbrace{\frac{\epsilon_0\,\norm{\Delta}_{L^2}}{\sqrt{T\mu_\Phi}}}_{\substack{\text{library bias}\\(\text{persistent})}}
  \;+\;
  \underbrace{\mathcal{O}_p\!\left(\sigma_\varepsilon\sqrt{\tfrac{p}{N\mu_\Phi}}\right)}_{\substack{\text{noise variance}\\(\to 0)}} ,
  \label{eq:ident_final}
\end{equation}
where $T$ is the record length. The variance term
vanishes as $N\to\infty$, whereas the library bias does not; hence
$\hat{\btheta}\to\btheta^*$ requires \emph{all three} of
$\epsilon_0\to0$ (rich library), $\sigma_\varepsilon\to0$, and
$N\to\infty$.
\end{theorem}

\begin{remark}[Scope of Theorem~\ref{thm:ident}]
The theorem concerns the coefficients of the \emph{projected}
hidden-state coupling within the observable library, up to the
persistent bias \eqref{eq:ident_bias}, for systems in the closure class.
It says nothing about the hidden trajectory $\bqh(t)$, and nothing about
systems whose hidden influence is history dependent. The model-based
reconstruction of Sect.~\ref{ssec:ekf_reconstruction} requires an
assumed model of the hidden dynamics and is not covered.
\end{remark}

\begin{proof}
The proof has three parts: existence/uniqueness, the deterministic
bias, and the stochastic variance.

\emph{Existence and uniqueness.}
Under (C2), $\lambda_{\min}(\bPhi_{\bxi}^{\mathsf{T}}\bPhi_{\bxi})
\ge N\mu_\Phi>0$, so the Gram matrix is symmetric positive definite and
invertible; the normal equations therefore have the unique solution
$\hat{\btheta}$ \cite{LehmannCasella1998}. Writing the sampled
observation model
\begin{equation}
  \ddtbqo=\bPhi_{\bxi}\btheta^*+\eta+\varepsilon ,
  \label{eq:ident_model}
\end{equation}
with $\eta$ the deterministic library mismatch of (C1) and $\varepsilon$
the zero-mean measurement noise, the estimation error splits into a
deterministic (bias) and a random (variance) part,
\begin{equation}
  \hat{\btheta}-\btheta^*
  =\underbrace{(\bPhi_{\bxi}^{\mathsf{T}}\bPhi_{\bxi})^{-1}\bPhi_{\bxi}^{\mathsf{T}}\eta}_{\text{bias}}
  +\underbrace{(\bPhi_{\bxi}^{\mathsf{T}}\bPhi_{\bxi})^{-1}\bPhi_{\bxi}^{\mathsf{T}}\varepsilon}_{\text{variance}} .
  \label{eq:ident_error}
\end{equation}

\emph{Bias.}
The spectral norm of the least-squares map is the reciprocal smallest
singular value of $\bPhi_{\bxi}$, bounded through (C2) by
\begin{equation}
  \norm{(\bPhi_{\bxi}^{\mathsf{T}}\bPhi_{\bxi})^{-1}\bPhi_{\bxi}^{\mathsf{T}}}_2
  =\lambda_{\min}(\bPhi_{\bxi}^{\mathsf{T}}\bPhi_{\bxi})^{-1/2}
  \le(N\mu_\Phi)^{-1/2}.
  \label{eq:ident_opnorm}
\end{equation}

Condition (C1) bounds the mismatch in the continuous $L^2$ norm, whereas
\eqref{eq:ident_opnorm} acts on the sampled (Euclidean) vector $\eta$;
the two are linked by the Riemann-sum identity
\begin{equation}
  \norm{\eta}_2^2=\sum_{k=1}^{N}\eta(t_k)^2
  =\frac{N}{T}\,\norm{\eta}_{L^2}^2\,\bigl(1+o(1)\bigr),
  \label{eq:ident_normbridge}
\end{equation}
so that $\norm{\eta}_2=\sqrt{N/T}\,\norm{\eta}_{L^2}(1+o(1))$. The
sampling factor $\sqrt{N}$ in \eqref{eq:ident_normbridge} cancels the
$\sqrt{N}$ in \eqref{eq:ident_opnorm}. With
$\norm{\eta}_{L^2}\le\epsilon_0\norm{\Delta}_{L^2}$ from (C1),
\begin{equation}
  \norm{(\bPhi_{\bxi}^{\mathsf{T}}\bPhi_{\bxi})^{-1}\bPhi_{\bxi}^{\mathsf{T}}\eta}_2
  \le\frac{\norm{\eta}_2}{\sqrt{N\mu_\Phi}}
  \le\frac{\epsilon_0\,\norm{\Delta}_{L^2}}{\sqrt{T\mu_\Phi}} .
  \label{eq:ident_bias}
\end{equation}

The bias is thus \emph{independent of $N$}: enlarging the record cannot
remove error caused by an incomplete library---only a richer library
($\epsilon_0\!\downarrow$) or stronger excitation ($\mu_\Phi\!\uparrow$)
can. (Writing the bias with the spurious factor $1/\sqrt{N}$, as would
follow from equating $\norm{\eta}_2$ with $\norm{\eta}_{L^2}$, is the one
subtlety the empirical-norm bridge \eqref{eq:ident_normbridge}
resolves.)

\emph{Variance.}
The noise term is zero-mean, so we bound it in mean square. For a fixed
matrix $\bm{M}$ and noise with
$\mathbb{E}[\varepsilon\varepsilon^{\mathsf{T}}]=\sigma_\varepsilon^2\bm{I}$,
the trace identity
\begin{equation}
  \mathbb{E}\norm{\bm{M}\varepsilon}_2^2
  =\sigma_\varepsilon^2\operatorname{tr}(\bm{M}\bm{M}^{\mathsf{T}}),
  \label{eq:ident_traceid}
\end{equation}
holds \cite{LehmannCasella1998}, where $\operatorname{tr}(\cdot)$ is the
matrix trace (sum of diagonal entries). Taking
$\bm{M}=(\bPhi_{\bxi}^{\mathsf{T}}\bPhi_{\bxi})^{-1}\bPhi_{\bxi}^{\mathsf{T}}$,
for which $\bm{M}\bm{M}^{\mathsf{T}}=(\bPhi_{\bxi}^{\mathsf{T}}\bPhi_{\bxi})^{-1}$,
\begin{equation}
  \mathbb{E}\norm{(\bPhi_{\bxi}^{\mathsf{T}}\bPhi_{\bxi})^{-1}\bPhi_{\bxi}^{\mathsf{T}}\varepsilon}_2^2
  =\sigma_\varepsilon^2\operatorname{tr}\!\bigl[(\bPhi_{\bxi}^{\mathsf{T}}\bPhi_{\bxi})^{-1}\bigr]
  \le\frac{\sigma_\varepsilon^2 p}{N\mu_\Phi},
  \label{eq:ident_var}
\end{equation}
the last bound using
$\operatorname{tr}[(\bPhi_{\bxi}^{\mathsf{T}}\bPhi_{\bxi})^{-1}]
=\sum_{\ell=1}^{p}\lambda_\ell^{-1}\le p\,\lambda_{\min}^{-1}$ and (C2).
By Markov's inequality this is
$\mathcal{O}_p(\sigma_\varepsilon\sqrt{p/(N\mu_\Phi)})$, which---unlike
the bias---vanishes as $N\to\infty$.

\emph{Combination.}
Inserting \eqref{eq:ident_bias} and \eqref{eq:ident_var} into
\eqref{eq:ident_error} and applying the triangle inequality gives
\eqref{eq:ident_final}. Under library mismatch ($\epsilon_0>0$) the
estimator converges to $\btheta^*$ plus the persistent bias
\eqref{eq:ident_bias}; only when the library spans the coupling exactly
($\epsilon_0=0$) does $\hat{\btheta}\overset{P}{\to}\btheta^*$ as
$N\to\infty$.
\end{proof}

\begin{remark}[Design insight]
Both terms in \eqref{eq:ident_final} scale as $\mu_\Phi^{-1/2}$, so
richer excitation (larger $\mu_\Phi$, e.g.\ a run-up) reduces bias and
variance \emph{together} --- experiment design, not sensor precision, is
the primary lever on accuracy.
\end{remark}

\begin{remark}[On the closure assumption]
Condition (C1) carries most of the weight of Theorem~\ref{thm:ident}:
it asserts that the hidden-state influence $\Delta$ is captured, up to a
relative remainder $\epsilon_0$, by a finite linear combination of
library terms evaluated at the instantaneous observable state. For the
simulated rotor of Sect.~\ref{sec:application} the full library captures
only 63\% ($x$-equation) and 76\% ($y$-equation) of the noise-free
closure energy (Sect.~\ref{ssec:capacity}), i.e.\ $\epsilon_0\approx0.5$--$0.6$;
the bias term of \eqref{eq:ident_final} is therefore not negligible
there, and identified coefficients are best read as projections of the
hidden coupling onto the library, not as physical constants.
\end{remark}

\subsection{Impossibility: when is recovery impossible?}
\label{ssec:impossibility}

Theorem~\ref{thm:ident} says when recovery is possible; the
complementary question is when it is provably impossible. If the
excitation is too poor for the candidate signatures to be linearly
independent---so the signature matrix loses column rank---then two
distinct coefficient vectors produce identical observable dynamics, and
no estimator can tell them apart, however long the record or precise the
sensors. This is the classical persistent-excitation (rank) condition of system
identification \cite{Ljung1999,SoderstromStoica1989}, specialized to the
observable signature matrix: a rank-deficient $\bPhi_{\bxi}$ leaves the
coefficients non-identifiable, and Le~Cam's two-point bound
\cite{Tsybakov2009} makes the impossibility estimator-independent. We
state it only to set up its concrete consequence for rotor
experiment design (Corollary~\ref{cor:const_speed}).

\begin{proposition}[Non-identifiability under insufficient excitation]
\label{thm:impossibility}
Suppose the signature matrix is rank-deficient over the observed
record, $\operatorname{rank}\bPhi_{\bxi} < p$ (equivalently
$\mu_\Phi = \lambda_{\min}(N^{-1}\bPhi_{\bxi}^{\mathsf{T}}
\bPhi_{\bxi}) = 0$). Then there exist $\btheta^{(1)} \neq \btheta^{(2)}$
that induce identical observable accelerations,
$\bPhi_{\bxi}\btheta^{(1)} = \bPhi_{\bxi}\btheta^{(2)}$. Consequently $\btheta^{(1)}$ and $\btheta^{(2)}$ are statistically
indistinguishable: whatever the record length $N$ or sensor precision
$\sigma_\varepsilon \geq 0$, no estimator $\hat{\btheta}$ can identify
which generated the data, and none is consistent.
\end{proposition}

\begin{proof}
Rank deficiency furnishes a nonzero null vector
$\bm{v}\in\mathrm{null}(\bPhi_{\bxi})$. Fixing the data-generating value
$\btheta^{(1)}=\btheta^*$ and setting $\btheta^{(2)}=\btheta^*+c\bm{v}$
with $c\neq0$ gives
\begin{equation}
  \bPhi_{\bxi}\btheta^{(2)}=\bPhi_{\bxi}\btheta^{(1)},
  \label{eq:imp_null}
\end{equation}
so both vectors produce the same mean observable acceleration at every
sample. Under the model $\ddtbqo=\bPhi_{\bxi}\btheta+\varepsilon$ with
$\varepsilon$ independent of $\btheta$, the induced data distributions
are equal, so their total-variation distance
$\mathrm{TV}(P,Q)\triangleq\sup_A|P(A)-Q(A)|$ vanishes,
$\mathrm{TV}(P_{\btheta^{(1)}},P_{\btheta^{(2)}})=0$, and the Le~Cam
two-point bound \cite{Tsybakov2009} gives, for any test
$\psi\in\{1,2\}$,
\begin{equation}
\begin{split}
  &\inf_{\psi}\;\max_{i\in\{1,2\}}P_{\btheta^{(i)}}(\psi\neq i)\\
  &\quad\ge\;\tfrac12\bigl(1-\mathrm{TV}(P_{\btheta^{(1)}},P_{\btheta^{(2)}})\bigr)
  =\tfrac12 .
\end{split}
  \label{eq:lecam}
\end{equation}

Any estimator induces such a test, so by \eqref{eq:lecam} it cannot
distinguish $\btheta^{(1)}$ from $\btheta^{(2)}$ better than chance; the
data-generating vector is therefore unrecoverable and no estimator is
consistent, uniformly in $N$ and $\sigma_\varepsilon$.
\end{proof}

\begin{remark}[What triggers rank deficiency]
The hypothesis is about the \emph{functional} rank of the library
along the trajectory---the number of linearly independent signatures
$\{\phi_\ell(\bxi(\cdot))\}$---not the geometric dimension of the state
manifold. A low-dimensional trajectory usually triggers it: the
constant-speed rotor (Corollary~\ref{cor:const_speed}) is the archetype,
where distinct candidates alias onto a common set of harmonics and
$\operatorname{rank}\bPhi_{\bxi}$ drops below $p$.
\end{remark}

\subsection{Gate consistency: does the algorithm recover the right model?}
\label{ssec:gate_consistency}

\emph{Does the Gate-1 screen detect true terms and control false
admissions?} The next proposition answers this for the idealized
regression; both ingredients, the Bonferroni bound \cite{Dunn1961} and
the $t$-test, are standard. It applies to the weak-form system
\eqref{eq:weak_system}, whose errors are independent across test
functions, with $N$ the number of test functions. Regressions whose
residuals are deterministic approximation errors, or pointwise
regressions on differentiated data with $N-p$ degrees of freedom, whose
errors are serially correlated, do not satisfy these assumptions;
Sect.~\ref{ssec:ablation} shows that tests based on them admit absent
mechanisms.

\begin{proposition}[Validity of the Gate-1 screen in the idealized regression]
\label{prop:validity}
Assume (C1)--(C2) and the error model of Theorem~\ref{thm:ident}, Gaussian errors, and
\begin{enumerate}[label=\textup{(C\arabic*)},start=3]
  \item Noise bound.
        $\mathrm{SNR} \triangleq
        \norm{\Delta}_{L^2}/\sigma_\varepsilon > \mathrm{SNR}_{\min}$,
        where
        \begin{equation}
          \mathrm{SNR}_{\min}
          = \frac{\sqrt{p}}{\sqrt{\mu_\Phi}\,(1 - \epsilon_0)}\,
            \sqrt{\frac{F_{0.05,\,p,\,N-p}}{N}}.
          \label{eq:SNR_min}
        \end{equation}
\end{enumerate}

Then \textup{(i)} each true term $\phi_\ell$ ($\theta_\ell^* \neq 0$)
passes Gate~1 with probability $\to 1$; and \textup{(ii)} each null term
($\theta_\ell^* = 0$) is rejected with family-wise error rate
$\leq \alpha$.
\end{proposition}

\begin{proof}
Throughout we work in the \emph{small-bias regime}
$\epsilon_0\sqrt{N}\to0$---exactly satisfied when the coupling lies in
the library span, $\Delta\in\mathrm{span}(\calL_{\mathrm{hidden}})$---so
that by Theorem~\ref{thm:ident} the persistent library bias
\eqref{eq:ident_bias} is negligible relative to the sampling standard
error and the estimator is effectively unbiased at the sampling scale.

\emph{(i) Detection of true terms.} Fix a true term,
$\theta_\ell^*\neq0$. In the small-bias regime Theorem~\ref{thm:ident}
gives $\hat\theta_\ell=\theta_\ell^*+\mathcal{O}_p(N^{-1/2})$, and, by
the coefficient covariance \eqref{eq:ident_var} and (C2), its standard
error obeys
\begin{equation}
  \hat\sigma_\ell
  =\sigma_\varepsilon\sqrt{\bigl[(\bPhi_{\bxi}^{\mathsf{T}}\bPhi_{\bxi})^{-1}\bigr]_{\ell\ell}}
  \;\le\;\frac{\sigma_\varepsilon}{\sqrt{N\mu_\Phi}} .
  \label{eq:gate1_se}
\end{equation}

The Gate-1 $t$-statistic is therefore
\begin{equation}
  |\hat t_\ell|=\frac{|\hat\theta_\ell|}{\hat\sigma_\ell}
  \;\ge\;\frac{|\theta_\ell^*|\sqrt{N\mu_\Phi}}{\sigma_\varepsilon}
  +\mathcal{O}_p(N^{-1/2}),
  \label{eq:cons_t}
\end{equation}
which grows without bound like $\sqrt N$. The Bonferroni critical value
$t_{\alpha/(2p),\,N-p}$ in \eqref{eq:gate1} grows only like
$\sqrt{2\log p}$, so once the signal exceeds the
threshold~\eqref{eq:SNR_min} the statistic overtakes the critical value
with probability $\to1$: every true term passes Gate~1.

\emph{(ii) Control of false discoveries.} Fix a null term,
$\theta_\ell^*=0$. In the small-bias regime its projected bias is
negligible, so $\hat\theta_\ell$ is pure estimation noise and, by
standard regression theory \cite{LehmannCasella1998,Montgomery2017}, its
studentized statistic follows the null distribution
$\hat t_\ell\sim t_{N-p}$. A two-sided test at the
Bonferroni-reduced tail level $\alpha/(2p)$ therefore falsely admits any
one null term with probability $\alpha/p$, and the union bound over the
at most $p$ candidates gives
\begin{equation}
  \mathbb{P}(\text{any null term admitted})
  \;\le\;\sum_{\ell:\,\theta_\ell^*=0}\frac{\alpha}{p}
  \;\le\;\alpha,
  \label{eq:cons_fwer}
\end{equation}
so the family-wise error rate is at most $\alpha$.
\end{proof}

\begin{remark}[Nominal and empirical error rates]
The family-wise level $\alpha$ in
Proposition~\ref{prop:validity}(ii) is nominal: it holds when the
regression errors are independent and the library bias is small. In the
rotor application the closure remainder is not small
($\epsilon_0\approx0.5$--$0.6$, Sect.~\ref{ssec:capacity}) and is
correlated across neighboring test functions, so the false-admission
rate of the weak-form test is assessed empirically by the knock-out
controls of Sect.~\ref{ssec:selection} (Table~\ref{tab:selection}).
\end{remark}

\begin{remark}[Termination]
With a finite library and a reserve set that is only ever enlarged, the
procedure of Algorithm~\ref{alg:spiral_po} stops after at most as many
passes as there are reserve groups. A stopping rule based on a Ljung--Box
whiteness test is not used, because its null hypothesis (white residual)
is violated by differentiated data even under a correct model.
\end{remark}

\subsection{Sample complexity: how much data is enough?}

The following proposition gives, for the idealized regression, the
number of independent equations needed for Gate~1 to detect a term of a
given size. It is a standard Wald-test power calculation
\cite{LehmannCasella1998} specialized to the Bonferroni threshold. Its
consequence $N_{\min}\propto\mu_\Phi^{-1}$ is a statement about this
idealized per-term problem only; numerical differentiation, regressor
errors, sequential selection and Gate~2 are not included, and $\mu_\Phi$
is itself a property of the realized trajectory and library, so
$N_{\min}$ is not an a priori sample-size requirement for the complete
procedure.

\begin{proposition}[Idealized per-term detection bound]
\label{thm:sample}
Under the assumptions of Proposition~\ref{prop:validity}, the Bonferroni
Gate~1 passes every true term $\phi_\ell$ (with
$|\theta_\ell^*| \geq \theta_{\min} > 0$) with probability at
least $1 - \delta$ provided
\begin{equation}
  N \;\geq\; N_{\min}(\delta)
  \;=\; \frac{\sigma_\varepsilon^2}{\mu_\Phi\,\theta_{\min}^2}
        \cdot
        \left[
          z_{1-\delta/2} + z_{1-\alpha/(2p)}
        \right]^2
        + p,
  \label{eq:Nmin}
\end{equation}
where $N$ counts independent equations (test functions for the weak
form), $z_\gamma = F_{\mathcal N}^{-1}(\gamma)$ is the standard normal
quantile, and
$\mu_\Phi = \lambda_{\min}(N^{-1}\bPhi_{\bxi}^{\mathsf{T}}
\bPhi_{\bxi})$.
\end{proposition}

\begin{proof}
Consider a single true coefficient with $|\theta_\ell^*|\ge\theta_{\min}$
in the small-bias regime of Proposition~\ref{prop:validity}, so that the
library bias \eqref{eq:ident_bias} is negligible against the standard
error. Then, by the coefficient covariance \eqref{eq:ident_var} and the
asymptotic normality of least squares \cite{LehmannCasella1998},
\begin{equation}
  \hat{\theta}_\ell\sim\mathcal{N}\!\bigl(\theta_\ell^*,\,
  \sigma_\varepsilon^2[(\bPhi_{\bxi}^{\mathsf{T}}\bPhi_{\bxi})^{-1}]_{\ell\ell}\bigr),
  \label{eq:sample_dist}
\end{equation}
and under (C2) its variance is controlled by the excitation quality,
$[(\bPhi_{\bxi}^{\mathsf{T}}\bPhi_{\bxi})^{-1}]_{\ell\ell}\le(N\mu_\Phi)^{-1}$,
so $\hat{\sigma}_\ell\le\sigma_\varepsilon/\sqrt{N\mu_\Phi}$. Gate~1
rejects $H_0:\theta_\ell=0$ when $|\hat t_\ell|>z_{1-\alpha/(2p)}$, the
Gaussian limit of the Bonferroni threshold. Under the alternative
$\theta_\ell=\theta_{\min}$ the statistic is centered at
$\theta_{\min}/\hat\sigma_\ell$, so its power reaches $1-\delta$ exactly
when this center clears the critical value by the $\delta$-quantile
margin,
\begin{equation}
  \frac{\theta_{\min}}{\hat{\sigma}_\ell}
  \;\ge\;z_{1-\alpha/(2p)}+z_{1-\delta/2}.
  \label{eq:sample_power}
\end{equation}

Substituting $\hat{\sigma}_\ell\le\sigma_\varepsilon/\sqrt{N\mu_\Phi}$ and
solving \eqref{eq:sample_power} for $N$ gives $N\ge N_{\min}(\delta)$ of
\eqref{eq:Nmin}, the additive $p$ accounting for the degrees of freedom
consumed by the fit.
\end{proof}

\begin{remark}[Scope of the bound]
Proposition~\ref{thm:sample} bounds the probability of detecting each
true term in the idealized regression; the probability of false
admissions is controlled separately by the family-wise bound of
Proposition~\ref{prop:validity}(ii) and by Gate~2.
\end{remark}

\section{Application: Vertical Flexible Rotor with Cubic Stiffness Nonlinearity} \label{sec:application}

SPIRAL-PO is evaluated on a partially observed 4-DOF rigid-disk rotor in
simulation, where the hidden tilt states and the presence or absence of
each mechanism are known. The laboratory rig, on which the spin speed
was not measured, is treated in Sect.~\ref{sec:testrig}.

\subsection{Partially Observed 4-DOF Rotor and Observable Subspace}
\label{ssec:system}
The 4-DOF vertical flexible rotor of \cite{Friswell2010} has
generalized coordinates
$\bq = \{x, y, \alpha, \beta\}^{\mathsf{T}}$, where $x(t)$, $y(t)$
are lateral disk displacements and $\alpha(t)$, $\beta(t)$ are the tilt
angles.
In an experimental setting, $x$ and $y$ are measured by non-contact
displacement sensors, while $\alpha$ and $\beta$ are not directly
accessible.
The observable and hidden sub-states are therefore
\begin{equation}
  \bqo = \begin{Bmatrix} x \\ y \end{Bmatrix},
  \qquad
  \bqh = \begin{Bmatrix} \alpha \\ \beta \end{Bmatrix}.
  \label{eq:partition_rotor}
\end{equation}

The scheduling parameter is the known spin speed $\OO(t)$.

Recovering the projected effect of $\bqh$ on the observable dynamics does
not imply a unique recovery of $\bqh(t)$ itself: multiple hidden-state
realizations can produce equivalent observable behavior, and the
discovered equations remain expressed entirely in observable quantities,
with no $\bqh$ term to solve for. SPIRAL-PO's identification step
(Sects.~\ref{ssec:ident_rotor}--\ref{ssec:hidden_state_validation}) is
therefore restricted to recovering an accurate observable closure model;
reconstructing $\bqh(t)$ itself is a separate, complementary task, carried
out afterward by a standard state estimator applied to the validated
dynamics (Sect.~\ref{ssec:ekf_reconstruction}), not by SPIRAL-PO itself.

\subsection{Four-DOF equations of motion}
\label{sec:4dof_eom}

The rotor is idealized as a rigid disk mounted at axial station $z=a$ on
a slender flexible shaft of length $L$ supported by two bearings (at
$z=0,L$) and driven at a prescribed spin speed $\OO(t)$; the
disk-location ratios are $\zeta=a/L$ and $\bar\zeta=1-\zeta$. This is the
four-degree-of-freedom vertical flexible-rotor model of Friswell et
al.~\cite{Friswell2010}, to which the reader is referred for the full
derivation; we recall only the elements the identification uses. The
mechanical idealization is shown in Fig.~\ref{fig:4DOF_rotor}.

\begin{figure}[H]
    \centering
    \includegraphics[width=1\linewidth]{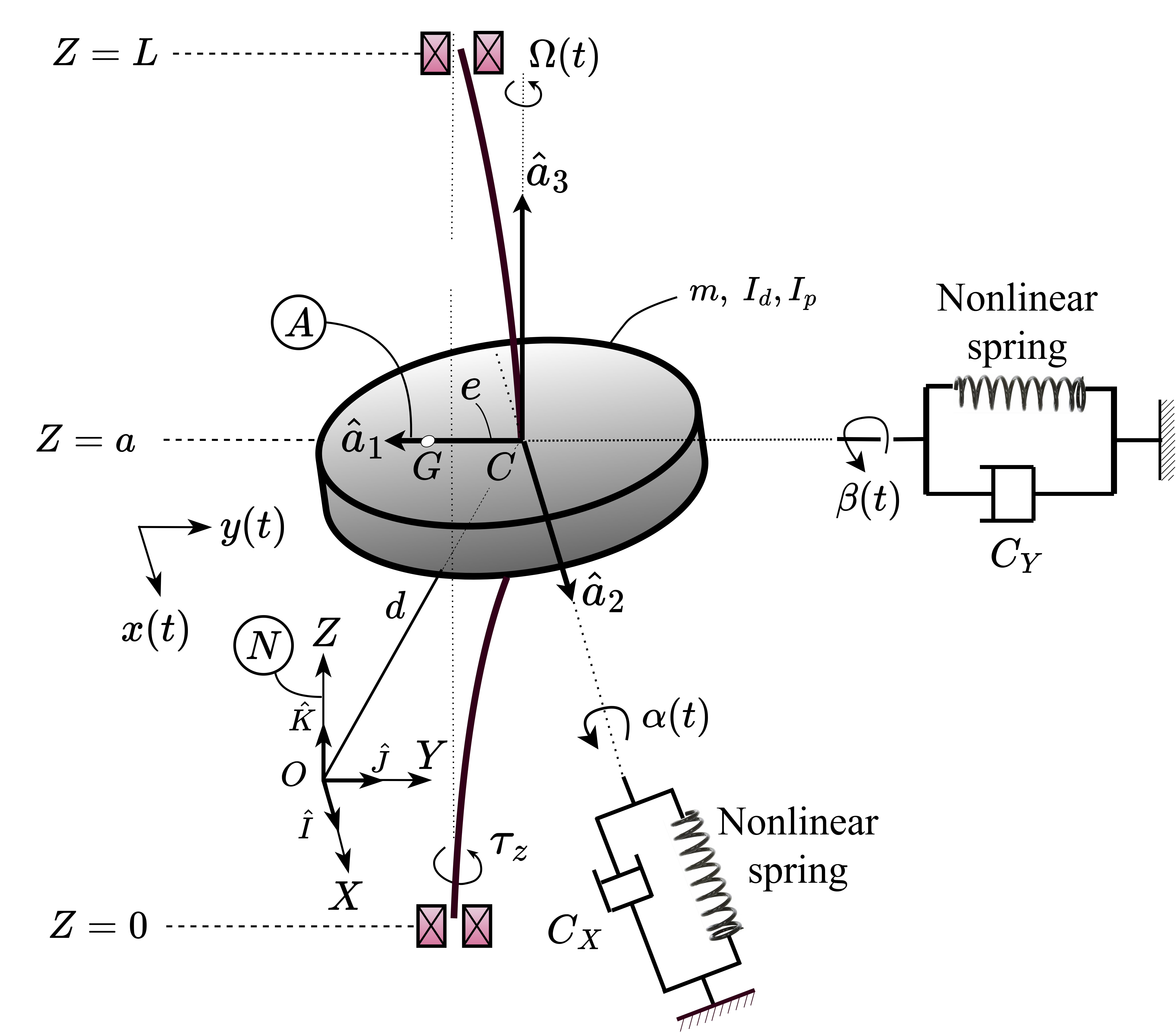}
    \caption{Schematic representation of the four-degree-of-freedom
    vertical rotor system.}
    \label{fig:4DOF_rotor}
\end{figure}

The disk has mass $m$, diametral and polar inertias $I_d$, $I_p$, and
mass eccentricity $e$; each bearing $i=1,2$ has a $2\times2$ linear
stiffness matrix $\bm{K}_{bi}$ and damping matrix $\bm{C}_{bi}$, both
with cross-axis terms, and a cubic (Duffing) stiffness $k_{NLi}$. Only
the lateral displacements $x,y$ are measured, while the tilts
$\alpha,\beta$ remain hidden. It can be shown that with the reduced slender-shaft approximation,
($\bm{K}_s\approx\bm{0}$), the equations of motion are \cite{Friswell2010}
\begin{equation}
  \bm{M}\,\ddot{\bq}
  + \bigl(\bm{C}_b + \OO\,\bm{G}\bigr)\,\dot{\bq}
  + \bigl[\bm{K}_L + \bm{K}_{NL}(\bq)\bigr]\,\bq
  = \bm{f}_{unb}(t),
  \label{eq:full_rotor}
\end{equation}
with $\bq = \{x,y,\alpha,\beta\}^{\mathsf{T}}$ as in
\eqref{eq:partition_rotor}. With disk mass $m$, diametral inertia $I_d$,
and polar inertia $I_p$, the mass and gyroscopic matrices are
\begin{equation}
\begin{aligned}
  \bm{M} &= \diag(m,\,m,\,I_d,\,I_d),\\[2pt]
  \bm{G} &=
  \begin{bmatrix}
    0&0&0&0\\ 0&0&0&0\\ 0&0&0&I_p\\ 0&0&-I_p&0
  \end{bmatrix},
\end{aligned}
  \label{eq:M_G_rotor}
\end{equation}
where $\bm{G}$ is skew-symmetric ($\bm{G}^{\mathsf{T}}=-\bm{G}$), so the
gyroscopic term $\OO\bm{G}\dot{\bq}$ contributes zero net power to the
system for any trajectory --- the standard structural requirement for a
physically admissible gyroscopic coupling. The linear
bearing stiffness $\bm{K}_L$ and bearing damping $\bm{C}_b$ are built
from the same influence-matrix projection as the nonlinear
stiffness \eqref{eq:KNL_rotor}:
\begin{equation}
\begin{aligned}
  \bm{K}_L &= \sum_{i=1}^{2} \bm{L}_i^{\mathsf{T}}\bm{K}_{bi}\bm{L}_i,\\[2pt]
  \bm{C}_b &= \sum_{i=1}^{2} \bm{L}_i^{\mathsf{T}}\bm{C}_{bi}\bm{L}_i,
\end{aligned}
  \label{eq:KL_Cb_rotor}
\end{equation}
where $\bm{K}_{bi}$, $\bm{C}_{bi}$ are the linear stiffness and damping
matrices of bearing $i$; their off-diagonal entries couple the $x$ and
$y$ directions at the bearing. Both the restoring force and the
dissipation originate at the same bearing locations, so
\eqref{eq:KL_Cb_rotor} completes the nonlinear stiffness
\eqref{eq:KNL_rotor}. $\bm{K}_L$ and $\bm{C}_b$ are positive definite
when $\bm{K}_{bi}$ and $\bm{C}_{bi}$ are, and $\bm{K}_{NL}(\bq)\bq$ is the gradient of the quartic potential
\begin{equation}
  V_{NL}(\bq) = \tfrac{1}{4}\sum_i k_{NLi}\norm{\bm{L}_i\bq}^4,
  \label{eq:VNL}
\end{equation}
so the conservative part of \eqref{eq:full_rotor} is a genuine mechanical
system, not an ad hoc combination of terms.
The Duffing nonlinear stiffness matrix is
\begin{equation}
  \bm{K}_{NL}(\bq)
  = \sum_{i=1}^{2} k_{NLi}\,\norm{\bm{L}_i\bq}^2\,
    \bm{L}_i^{\mathsf{T}}\bm{L}_i.
  \label{eq:KNL_rotor}
\end{equation}

Each bearing is loaded by the lateral displacement of the shaft
centerline at its own axial station, not by the disk displacement
$x,y$. A small tilt $\beta$ (about $Oy$) or $\alpha$ (about $Ox$) adds a
lever-arm contribution to the disk translation, so the centerline
displacements at bearing~1 (axial distance $a=\zeta L$ from the disk) and
bearing~2 (distance $b=\bar\zeta L$) are
\begin{equation}
\begin{aligned}
  \bm{u}_1 &= \begin{Bmatrix} x + \zeta L\,\beta \\[2pt] y - \zeta L\,\alpha \end{Bmatrix},\\[2pt]
  \bm{u}_2 &= \begin{Bmatrix} x - \bar\zeta L\,\beta \\[2pt] y + \bar\zeta L\,\alpha \end{Bmatrix}.
\end{aligned}
  \label{eq:bearing_kin}
\end{equation}

In the reduced flexible-shaft model of \cite{Friswell2010}, each bearing
reacts the disk load in proportion to its static participation in the
simply-supported span --- fraction $\bar\zeta$ at bearing~1, $\zeta$ at
bearing~2 --- so the constitutive displacement is
$\bm{r}_1 = \bar\zeta\,\bm{u}_1$, $\bm{r}_2 = \zeta\,\bm{u}_2$. Writing
$\bm{r}_i = \bm{L}_i\bq$ combines the rigid-shaft lever arms of
\eqref{eq:bearing_kin} with these participation weights to give the
influence matrices $\bm{L}_i \in \R^{2\times 4}$
($\zeta = a/L$, $\bar\zeta = 1 - \zeta$):
\begin{equation}
\begin{aligned}
  \bm{L}_1 &=
  \begin{bmatrix}
    \bar\zeta & 0 & 0 & \zeta\bar\zeta L \\
    0 & \bar\zeta & -\zeta\bar\zeta L & 0
  \end{bmatrix},\\[2pt]
  \bm{L}_2 &=
  \begin{bmatrix}
    \zeta & 0 & 0 & -\zeta\bar\zeta L \\
    0 & \zeta & \zeta\bar\zeta L & 0
  \end{bmatrix}.
\end{aligned}
  \label{eq:Li_rotor}
\end{equation}

Forcing is translational only,
\begin{equation}
  \bm{f}_{unb}(t) = me\begin{Bmatrix}
    -\OO^2\cos\psi-\dot\OO\sin\psi\\ -\OO^2\sin\psi+\dot\OO\cos\psi\\ 0\\ 0
  \end{Bmatrix},
  \label{eq:f_unb}
\end{equation}
with $\psi=\vartheta(t)+\phi_0$, where $\vartheta(t)=\int_0^t\OO(\tau)\,d\tau$
is the accumulated spin angle (equal to $\OO t$ only at constant speed).

\subsection{Observable equations of motion}

Let $\br_i=\bL_i\bq$ be the bearing-plane displacements of
\eqref{eq:Li_rotor} and
\begin{equation}
  \bm{F}_i=\bm{C}_{bi}\dot{\br}_i+\bm{K}_{bi}\br_i
  +k_{NLi}\norm{\br_i}^2\br_i
  =\begin{Bmatrix}F_{ix}\\ F_{iy}\end{Bmatrix}
  \label{eq:bearing_force}
\end{equation}
the reaction of bearing $i$. Extracting the $x$- and $y$-rows of
\eqref{eq:full_rotor}, with the bearing weights $w_1=\bar\zeta$,
$w_2=\zeta$,
\begin{align}
  m\ddot{x} &= f_{unb,x}(t)
  \underbrace{- \sum_{i=1}^{2} w_i F_{ix}}_{\text{contains hidden }\alpha,\beta},
  \label{eq:xobs} \\[6pt]
  m\ddot{y} &= f_{unb,y}(t)
  \underbrace{- \sum_{i=1}^{2} w_i F_{iy}}_{\text{contains hidden }\alpha,\beta},
  \label{eq:yobs}
\end{align}
where $r_{1x}=\bar\zeta(x+\zeta L\beta)$, $r_{2x}=\zeta(x-\bar\zeta L\beta)$,
$r_{1y}=\bar\zeta(y-\zeta L\alpha)$, $r_{2y}=\zeta(y+\bar\zeta L\alpha)$.

For reference, the two (unmeasured) tilt rows of \eqref{eq:full_rotor} are
\begin{align}
  I_d\ddot{\alpha} &= -\,\OO I_p\,\dot{\beta}
  +\zeta\bar\zeta L\sum_{i=1}^{2}(-1)^{i-1}F_{iy},
  \label{eq:alphaobs} \\[6pt]
  I_d\ddot{\beta} &= +\,\OO I_p\,\dot{\alpha}
  -\zeta\bar\zeta L\sum_{i=1}^{2}(-1)^{i-1}F_{ix},
  \label{eq:betaobs}
\end{align}
where $(-1)^{i-1}$ encodes the opposite lever-arm sense of the two
bearings. The hidden states enter the observable rows through
(i) the linear bearing coupling, i.e.\ the tilt contributions to
$r_{ix},r_{iy}$ and their rates, including through the cross-axis
entries of $\bm{K}_{bi}$ and $\bm{C}_{bi}$, and (ii) the cubic amplitude
$\norm{\br_i}^2$, itself a function of $\alpha,\beta$
through $\bm{L}_i\bq$. Conversely, the gyroscopic terms
$\OO I_p\dot\beta$, $\OO I_p\dot\alpha$ in
\eqref{eq:alphaobs}--\eqref{eq:betaobs} couple the two tilt channels, and
this coupling --- transmitted back through the bearing reactions --- is
what produces the $\OO\dot y$, $\OO\dot x$ gyroscopic signatures
SPIRAL-PO must recover from $x,y$ alone.

\subsection{Observable regressor and mechanism-labeled library}
\label{ssec:rotor_library}

The observable regressor is
\begin{equation}
  \bxi(t)
  = \{ x,\; y,\; \dot{x},\; \dot{y},\; \OO,\; \dot\OO\}^{\mathsf{T}} ,
  \label{eq:regressor_rotor}
\end{equation}
with the spin speed and its rate known. No tilt angle appears. The
candidates follow from projecting the hidden-coupling terms of
\eqref{eq:xobs}--\eqref{eq:yobs} onto $x,y$ (Table~\ref{tab:CandidateLibrary}):
the gyroscopic signature $\OO\dot y$ in the $x$-equation ($\OO\dot x$ in
the $y$-equation), which arises when $\OO I_p\dot\beta$, $\OO
I_p\dot\alpha$ are transmitted through the bearing reactions; the
isotropic cubic term $(x^2+y^2)x$ ($(x^2+y^2)y$), which follows from the
small-tilt approximation $\norm{\br_i}^2\approx w_i^2(x^2+y^2)$; a
cross-channel term $y$ ($x$) for linear coupling between the two
directions; and
alternative or higher-order forms ($(x^2+y^2)\dot x$, $x^3$, $x^2y$,
$\OO^2x$, $\dot\OO x$ and their $y$-equation counterparts) that compete
with the physical terms and therefore test the selection. The bearings
of this benchmark have cross-axis stiffness and damping
(Appendix~\ref{app:synthetic_matrices}), so the cross-channel term is
physically present, through the bearings as well as through the tilt
dynamics.

\begin{table*}[!t]
\caption{Candidate library for the rotor, per observable equation. The
seed regressors are fitted in every model; the eight candidates are
tested jointly. Candidates collinear with the seed regressors on a given
record are removed before testing (Sect.~\ref{ssec:library}).}
\label{tab:CandidateLibrary}
\centering
\tablefont
\begin{tabular}{p{2.30cm} p{3.70cm} p{3.70cm} p{4.40cm}}
\toprule
Class & $x$-equation & $y$-equation & Mechanism \\
\midrule
Seed (always present) & $u_c,\ u_s,\ \dot x,\ x$ (Eq.~\eqref{eq:seed_rotor})
  & $u_c,\ u_s,\ \dot y,\ y$
  & unbalance forcing, linear damping and stiffness \\
Gyroscopic & $\OO\dot y$ & $\OO\dot x$ & tilt gyroscopic coupling transmitted by the bearings \\
Cross-channel & $y$ & $x$ & linear $x$--$y$ coupling (bearing cross-stiffness, tilt) \\
Isotropic cubic & $(x^2+y^2)x$ & $(x^2+y^2)y$ & cubic bearing stiffness \\
Competing forms & $(x^2+y^2)\dot x,\ x^3,\ x^2y,\ \OO^2x,\ \dot\OO x$
  & $(x^2+y^2)\dot y,\ y^3,\ xy^2,\ \OO^2y,\ \dot\OO y$
  & nonlinear damping, single-axis or mixed cubic, speed-dependent stiffness \\
\bottomrule
\end{tabular}
\end{table*}

\subsection{Physics seed for the rotor}
\label{ssec:seed_rotor}

The seed is the linear planar oscillator with unbalance forcing,
\begin{equation}
\label{eq:seed_rotor}
\begin{split}
\fphys^{(0)}(\bxi) &= \bigl\{f^{(0)}_x,\ f^{(0)}_y\bigr\}^{\mathsf T},\\
f^{(0)}_x &= a_1u_c + a_2u_s - c_x\dot{x} - k_x x,\\
f^{(0)}_y &= b_1u_c + b_2u_s - c_y\dot{y} - k_y y,\\
u_c &= \OO^{2}\cos\vartheta+\dot\OO\sin\vartheta,\\
u_s &= \OO^{2}\sin\vartheta-\dot\OO\cos\vartheta,
\end{split}
\end{equation}
where $u_c$, $u_s$ are the components of the unbalance force
\eqref{eq:f_unb} at variable speed, and the forcing amplitude and phase
($a_1,a_2$, $b_1,b_2$), damping
and stiffness are estimated jointly with the candidates. No gyroscopic,
cubic or cross-channel term is present.

\subsection{Identifiability of the rotor}
\label{ssec:ident_rotor}

\begin{corollary}[Synchronous constant-speed response is rank deficient]
\label{cor:const_speed}
Suppose the response at constant speed $\OO_0$ is purely synchronous,
\begin{equation}
  x = A_x\cos(\OO_0 t+\varphi_x),\qquad y = A_y\sin(\OO_0 t+\varphi_y).
  \label{eq:const_traj}
\end{equation}

Then every signature of Table~\ref{tab:CandidateLibrary} lies in the
span of $\{\cos k\OO_0t,\sin k\OO_0 t\}_{k=1,3}$: the linear, gyroscopic
and cross-channel terms are first-harmonic, the cubic terms contain the
first and third harmonics, $\OO_0^2x$ is a multiple of $x$, and $\dot\OO
x\equiv0$. Hence $\operatorname{rank}\bPhi_{\bxi}\le4$. With more than
four columns Proposition~\ref{thm:impossibility} applies: the
gyroscopic, cubic and cross-channel coefficients cannot be separated by
any estimator from such data.
\end{corollary}

The corollary concerns the synchronous steady-state response. Any excitation
that lifts the harmonic degeneracy can restore rank: free or transient
response components, subharmonic or amplitude-modulated response,
broadband forcing (on the rig, through the lateral actuators), or
several constant-speed runs at different unbalance configurations. A
speed sweep is one convenient way to do it, and it does not guarantee
good conditioning by itself, which depends on the speed range, the
deceleration rate, the response amplitude, the library and the noise.

We checked the corollary on a synthetic constant-speed record with
simultaneous $x$ and $y$ (30\,Hz for 30\,s, last 5\,s used, after the
free transient has decayed; seed regressors plus the six candidates that
remain non-degenerate at constant speed). The column-normalized
signature matrix has four singular values of order one
($2.24$, $2.11$, $0.66$, $0.29$) and six at the level set by the 2\%
noise ($0.039$ to $0.006$), as the corollary predicts.

On the synthetic coast-down (Sect.~\ref{ssec:hidden_state_validation})
the column-normalized signature matrix of the seed plus all eight
candidates has condition number 113 ($x$-equation) and 108
($y$-equation), with $\mu_\Phi=\lambda_{\min}=3.4\times10^{-4}$ and
$3.8\times10^{-4}$. Its singular values decrease from $2.09$ to $0.18$
($x$) and from $2.11$ to $0.22$ ($y$), followed by one smaller value
($0.019$ in both equations) that reflects the near-collinearity of
$\dot\OO x$ with the seed column $x$ under a constant deceleration; all
twelve columns are numerically independent. For the cluster
representatives (correlation below 0.7) the condition numbers become 6.9
and 7.3, with $\mu_\Phi=0.061$ and $0.056$.

The idealized bound of Proposition~\ref{thm:sample} gives, for this
library, a sense of scale rather than a requirement.

\begin{corollary}[Idealized detection bound for the rotor]
\label{cor:sample_rotor}
With an equation-error standard deviation equal to $5\%$ of the
signature RMS, $K=8$ candidates, $\theta_{\min}=0.1\,\theta_{\mathrm{true}}$,
$\alpha=\delta=0.05$, and the measured $\mu_\Phi\approx0.06$,
\begin{equation}
\begin{split}
  N_{\min}
  &\approx \frac{(0.05)^2}{\mu_\Phi\,(0.1)^2}
  \bigl[z_{0.975}+z_{1-0.05/16}\bigr]^2\\
  &\approx \frac{0.25\times 22.0}{0.06}\approx 92
\end{split}
  \label{eq:Nmin_rotor}
\end{equation}
independent equations, well below the roughly 290 disjoint test
functions of the 15\,s benchmark record.
\end{corollary}

\subsection{Synthetic benchmark and study design}
\label{ssec:hidden_state_validation}

The 4-DOF model \eqref{eq:full_rotor}--\eqref{eq:KL_Cb_rotor} is
integrated with a fourth-order Runge--Kutta scheme (output step 1\,ms,
four substeps) over a 15\,s linear coast-down from 40 to 0\,Hz
(Table~\ref{tab:sim_params} and Appendix~\ref{app:synthetic_matrices};
one configuration is used throughout, and the released code builds every
record from it). The tilt angles are
integrated as genuine coupled states and then discarded; only $x(t)$,
$y(t)$ with 2\% white noise (per channel, relative to the channel
standard deviation) enter the identification. The linear natural
frequencies at rest are 13.8 and 16.1\,Hz (translation) and 141.0 and
164.5\,Hz (tilt). The coast-down passes the translational resonances
near 9--10\,s, where the disk displacement reaches 0.28\,mm and the tilt
0.37\,mrad and the cubic bearing force reaches about 30\% of the linear
one. The conservative part of the model
conserves the mechanical energy
\begin{equation}
  E = \tfrac12\dot{\bq}^{\mathsf T}\bm{M}\dot{\bq}
      + \tfrac12\bq^{\mathsf T}\bm{K}_L\bq + V_{NL}(\bq)
  \label{eq:energy}
\end{equation}
to a relative drift below $10^{-9}$ over $0.2\,$s.

\begin{table}[!t]
\caption{Nominal synthetic rotor benchmark and simulation settings
(code: \texttt{spiral\_config.m}, \texttt{simulate\_rotor.m}).}
\label{tab:sim_params}
\centering
\tablefont
\begin{tabular}{@{}>{\raggedright\arraybackslash}p{0.38\columnwidth}>{\raggedright\arraybackslash}p{0.58\columnwidth}@{}}
\toprule
Parameter & Value \\
\midrule
Disk mass $m$; radius $R$ & $1.15$\,kg; $0.0735$\,m \\
Inertias $I_p=mR^2/2$, $I_d=I_p/2$ & $3.106\times10^{-3}$, $1.553\times10^{-3}$\,kg\,m$^2$ \\
Span $L$; disk location $\zeta=a/L$ & $1.0$\,m; $0.25$ \\
Bearing stiffness $\bm K_{bi}$ (each) & $k_{xx}=1.860\times10^{4}$, $k_{yy}=2.225\times10^{4}$, $k_{xy}=k_{yx}=2.542\times10^{3}$\,N/m \\
Bearing damping $\bm C_{bi}$ (each) & $c_{xx}=8.0$, $c_{yy}=4.8$, $c_{xy}=c_{yx}=0.16$\,N\,s/m \\
Cubic stiffness $k_{NL1},k_{NL2}$ & $3.122\times10^{11}$\,N/m$^3$ each \\
Zero-speed undamped natural frequencies & $13.8$, $16.1$, $141.0$, $164.5$\,Hz \\
Unbalance $U=me$; nominal phase $\phi_0$ & $1.8\times10^{-5}$\,kg\,m; $0.4$\,rad \\
Prescribed truth speed & Linear, $40\rightarrow0$\,Hz over $15$\,s; $\dot{\OO}=-16.76$\,rad/s$^2$ \\
Truth integration & RK4; output interval $1$\,ms; four substeps of $0.25$\,ms each \\
Truth initial state & $\bm z(0)=\bm 0\in\mathbb R^8$ \\
Displacement measurement noise & Additive Gaussian white noise; standard deviation $2\%$ of each noise-free $x,y$ channel's standard deviation \\
Noise seeds & 31, 41, 51 (identification); 101--120 (Monte Carlo) \\
Independent test record & Noise seed $77$; phase $\phi_0+0.7=1.1$\,rad \\
Displacement derivatives & Savitzky--Golay, polynomial order $4$; window by cross-validation from $\{21,31,51\}$ samples \\
Processed speed signal & Samples spaced $50$\,ms apart; seven-point moving average and linear interpolation \\
Processed speed derivative & Central differences; one-sided differences at the endpoints \\
Weak form & $h=25$\,ms, disjoint supports ($\approx290$ test functions) \\
\bottomrule
\end{tabular}
\end{table}

Three kinds of experiment are reported. In the nominal benchmark all
three hidden-state mechanisms are present. In two knock-out controls a
mechanism is removed from the data-generating model: the cubic bearing
stiffness ($k_{NL1}=k_{NL2}=0$) or the gyroscopic matrix ($\bm G=\bm0$);
a correct method must not admit the corresponding terms. In four
structural-mismatch cases the data-generating model differs from the
assumed structure (Sect.~\ref{ssec:mismatch}). Every identified model is
evaluated on an independent record with a different noise realization
and unbalance phase, through the share of the post-seed acceleration
residual it explains there, with only the unbalance amplitude and phase
refitted on that record (the \emph{held-out gain}), and through
$k$-step-ahead prediction of $x,y$. Results were computed in GNU Octave
8.4 with the released code (\texttt{study\_synthetic.m}).

\subsection{Selection in the nominal benchmark and the controls}
\label{ssec:selection}

Table~\ref{tab:selection} gives, for each method, how often each
mechanism is selected over 20 noise realizations of the same trajectory
(10 for the gated network formulation, whose networks make each run slower).
The weak-form selection recovers the gyroscopic, cubic and cross-channel
terms in every nominal realization and selects about three terms in the
$x$-equation and four in the $y$-equation; the additional term is mostly
$\dot\OO y$, which is nearly collinear with the seed column $y$ under a
constant deceleration. With the cubic stiffness removed it admits a
cubic term in none of the 40 equation fits. With the gyroscopic matrix
removed it still admits the gyroscopic term in 15\% ($x$) and 25\% ($y$)
of the realizations, more often than the nominal 5\% level: in this
benchmark the library remainder is large ($\epsilon_0\approx0.5$--$0.6$,
Sect.~\ref{ssec:capacity}) and correlated across test functions, so the
nominal level is not reached, as anticipated in Sect.~\ref{sec:identifiability}.
Pointwise OLS with classical standard errors selects about four terms per
equation and admits the cubic term in 45--60\% of the no-cubic
realizations, because it treats about 15{,}000 serially correlated
samples as independent; spacing the samples one derivative window apart
reduces false admissions but loses power (the gyroscopic term is found in
0--40\% of the nominal realizations). The gated network formulation
selects seven to eight of the eight candidates in every case, including
the removed mechanisms.

\begin{table*}[!t]
\caption{Selection rates of the gyroscopic ($\OO\dot y$ / $\OO\dot x$),
isotropic cubic ($(x^2+y^2)x$ / $(x^2+y^2)y$) and cross-channel ($y$ /
$x$) terms over independent noise realizations (2\% noise), and mean
number of selected candidates. Bonferroni level $\alpha=0.05$ over eight
candidates. Entries in bold are mechanisms removed from the data-generating model (any non-zero rate is a false admission); the cross-channel term is present in all cases.}
\label{tab:selection}
\centering
\tablefont
\begin{tabular}{l l cccc cccc}
\toprule
 & & \multicolumn{4}{c}{$x$-equation} & \multicolumn{4}{c}{$y$-equation} \\
\cmidrule(lr){3-6}\cmidrule(lr){7-10}
Data & Method & gyro & cubic & cross & \# & gyro & cubic & cross & \# \\
\midrule
Nominal & Weak form & 1.00 & 1.00 & 1.00 & 3.05 & 1.00 & 1.00 & 1.00 & 3.85 \\
 & Pointwise OLS, classical SE & 1.00 & 1.00 & 1.00 & 4.15 & 1.00 & 1.00 & 1.00 & 4.00 \\
 & Pointwise OLS, spaced samples & 0.40 & 1.00 & 1.00 & 2.80 & 0.00 & 1.00 & 1.00 & 2.05 \\
 & Gated network & 1.00 & 1.00 & 1.00 & 7.2 & 1.00 & 1.00 & 1.00 & 7.9 \\
\midrule
No cubic & Weak form & 1.00 & \textbf{0.00} & 1.00 & 2.10 & 1.00 & \textbf{0.00} & 1.00 & 2.30 \\
 & Pointwise OLS, classical SE & 1.00 & \textbf{0.45} & 1.00 & 3.75 & 1.00 & \textbf{0.60} & 1.00 & 4.15 \\
 & Pointwise OLS, spaced samples & 0.60 & \textbf{0.00} & 1.00 & 1.65 & 0.30 & \textbf{0.45} & 1.00 & 1.85 \\
 & Gated network & 1.00 & \textbf{0.90} & 1.00 & 7.4 & 1.00 & \textbf{1.00} & 1.00 & 6.8 \\
\midrule
No gyroscopic & Weak form & \textbf{0.15} & 1.00 & 1.00 & 2.25 & \textbf{0.25} & 1.00 & 1.00 & 3.10 \\
 & Pointwise OLS, classical SE & \textbf{0.10} & 1.00 & 1.00 & 3.30 & \textbf{0.00} & 1.00 & 1.00 & 3.05 \\
 & Pointwise OLS, spaced samples & \textbf{0.15} & 1.00 & 1.00 & 2.55 & \textbf{0.00} & 1.00 & 1.00 & 2.10 \\
 & Gated network & \textbf{1.00} & 1.00 & 1.00 & 8.0 & \textbf{1.00} & 1.00 & 1.00 & 7.9 \\
\bottomrule
\end{tabular}
\end{table*}

Over noise levels from 0.5\% to 10\% (five realizations per level) the
cubic and cross-channel terms are selected at every level and the cubic
term is never admitted in the no-cubic control; the gyroscopic term is
selected in all nominal realizations up to 5\% noise and in 60\% at
10\%. In the gyroscopic-free control the gyroscopic term is admitted in
0--20\% of the realizations at 2--10\% noise, but in all of them at
0.5\%: when the noise is small, the library remainder $\eta$ of
\eqref{eq:closure_class} rather than the noise dominates the errors, and
the test detects the closure bias. This is the limit of
Theorem~\ref{thm:ident} that the bias does not vanish with better data.

Table~\ref{tab:coefficients} lists the identified coefficients for the
nominal record with both uncertainty estimates. The block-bootstrap
standard errors of the pointwise regression are two to twelve times the
classical OLS ones, consistent with the effective sample size being
roughly the number of samples divided by the derivative window. The
weak-form and pointwise coefficients differ by more than their
uncertainties, which is expected when $\epsilon_0$ is large: the two
estimators weight the library remainder differently, and neither
coefficient should be read as a physical bearing constant.

\begin{table*}[!t]
\caption{Identified model for the nominal record (noise seed 31): seed
and selected terms ($u_c$, $u_s$ as in \eqref{eq:seed_rotor}). Weak form: coefficient, standard error and $t$; pointwise
trajectory regression: coefficient and moving-block-bootstrap standard
error ($B=1000$, block length 84 samples). The weak-form fit uses about
290 test functions and 12 parameters per equation; terms are admitted
for $|t|>2.73$ (Bonferroni, $K=8$).}
\label{tab:coefficients}
\centering
\tablefont
\begin{tabular}{l l r r r r r}
\toprule
Eq. & Term & weak coeff. & SE & $t$ & pointwise coeff. & bootstrap SE \\
\midrule
$x$ & $u_c$ & $-1.37\times10^{-5}$ & $2.0\times10^{-8}$ & $-673$ & $-1.32\times10^{-5}$ & $3.0\times10^{-8}$ \\
    & $u_s$ & $7.40\times10^{-6}$ & $1.9\times10^{-8}$ & 395 & $7.11\times10^{-6}$ & $2.8\times10^{-8}$ \\
    & $-\dot x$ & $3.47$ & $0.033$ & 104 & $3.38$ & $0.061$ \\
    & $-x$ & $8.06\times10^{3}$ & $8.2$ & 985 & $8.18\times10^{3}$ & $21$ \\
    & $\OO\dot y$ & $-5.37\times10^{-3}$ & $3.9\times10^{-4}$ & $-13.7$ & $-6.67\times10^{-3}$ & $8.5\times10^{-4}$ \\
    & $y$ & $-1.10\times10^{3}$ & $3.9$ & $-281$ & $-1.11\times10^{3}$ & $7.0$ \\
    & $(x^2+y^2)x$ & $-3.38\times10^{10}$ & $1.8\times10^{8}$ & $-185$ & $-3.11\times10^{10}$ & $4.0\times10^{8}$ \\
\midrule
$y$ & $u_c$ & $-7.41\times10^{-6}$ & $1.8\times10^{-8}$ & $-422$ & $-7.09\times10^{-6}$ & $2.2\times10^{-8}$ \\
    & $u_s$ & $-1.38\times10^{-5}$ & $1.8\times10^{-8}$ & $-772$ & $-1.32\times10^{-5}$ & $3.1\times10^{-8}$ \\
    & $-\dot y$ & $2.06$ & $0.035$ & 59 & $1.96$ & $0.058$ \\
    & $-y$ & $9.64\times10^{3}$ & $7.7$ & 1260 & $9.75\times10^{3}$ & $19$ \\
    & $\OO\dot x$ & $4.15\times10^{-3}$ & $3.3\times10^{-4}$ & 12.7 & $3.94\times10^{-3}$ & $5.8\times10^{-4}$ \\
    & $x$ & $-1.10\times10^{3}$ & $2.9$ & $-377$ & $-1.09\times10^{3}$ & $6.2$ \\
    & $(x^2+y^2)y$ & $-3.38\times10^{10}$ & $1.9\times10^{8}$ & $-178$ & $-3.09\times10^{10}$ & $4.7\times10^{8}$ \\
\bottomrule
\end{tabular}
\end{table*}

The signs are physically consistent: the cubic coefficients are negative
(hardening) in both equations, and the gyroscopic coefficients have
opposite signs in the two equations, as the skew-symmetric gyroscopic
matrix requires. The cross-channel coefficients are equal in the two
equations ($-1.10\times10^{3}$\,s$^{-2}$), as the symmetric bearing
cross-stiffness requires, and of the order of
$-K_{L,12}/m=-1.38\times10^{3}$\,s$^{-2}$ from Appendix~\ref{app:synthetic_matrices}.
The forcing coefficients reproduce the unbalance amplitude
($\sqrt{a_1^2+a_2^2}=1.56\times10^{-5}$\,m against $U/m=1.57\times10^{-5}$\,m)
and, approximately, its phase ($0.49$\,rad in the $x$-equation against
$\phi_0=0.4$\,rad).

\subsection{Ablations and a delay-coordinate baseline}
\label{ssec:ablation}

Table~\ref{tab:ablation} isolates each component on three noise
realizations, all evaluated on the same independent record. Three
results stand out. First, every variant of the gated loop that uses valid
statistics gives \emph{identical} selections and held-out gains whether
the regression target is the measured residual itself, the network
evaluated on the trajectory, the network on a tube around it (radius
0.05--0.3 standard deviations) or the network on the Halton box: the
residual network and the Halton projection have no measurable effect.
Second, that loop is conservative: with time-blocked validation it
admits only the cross-channel term in each equation, whereas the
weak-form selection also finds the gyroscopic and cubic terms. Third,
the cross-channel term carries most of the predictive value of the
closure: on the independent record it explains about 40\% ($x$) and
55\% ($y$) of the post-seed residual, and the gyroscopic and cubic terms
add about 0.3 percentage points; models with more terms (the gated
network formulation, or no gates) gain less than one further point while
admitting terms absent from the data. Every closure model reduces the
10-ms prediction error from 0.062 / 0.071 (seed alone) to about
0.044 / 0.040.

\begin{table*}[!t]
\caption{Ablation study on the nominal benchmark (noise seeds 31, 41, 51;
values agree across seeds to the digits shown). Held-out gain: share of the post-seed acceleration
residual of an independent record explained by the selected terms
(unbalance amplitude and phase refitted on that record).
NRMSE: $k$-step prediction error of $x$ / $y$ divided by the channel
standard deviation, $k=10$ (10\,ms) and $k=50$ (50\,ms), from 300
starting points.}
\label{tab:ablation}
\centering
\tablefont
\begin{tabular}{p{4.04cm} p{3.76cm} p{2.26cm} p{1.79cm} p{1.88cm}}
\toprule
Configuration & Selected ($x$ / $y$) & held-out gain $x$ / $y$ (\%) & NRMSE $k{=}10$ & NRMSE $k{=}50$ \\
\midrule
Seed only & -- / -- & 0 / 0 & 0.062 / 0.071 & 0.29 / 0.35 \\
Weak-form selection (SPIRAL-PO) & $\OO\dot y$, $y$, $(x^2{+}y^2)x$ / $\OO\dot x$, $x$, $(x^2{+}y^2)y$ & 40.6 / 55.7 & 0.044 / 0.041 & 0.15 / 0.15 \\
Gated loop, valid statistics, NN on tube & $y$ / $x$ & 40.2 / 55.4 & 0.044 / 0.040 & 0.15 / 0.15 \\
\quad same, no NN (measured residual) & identical & identical & identical & identical \\
\quad same, NN on trajectory / on Halton box & identical & identical & identical & identical \\
\quad same, no statistical gates (STLSQ) & 7 terms / 6--7 terms & 41.0 / 55.7 & 0.044 / 0.041 & 0.15 / 0.15 \\
\quad same, no physics seed & seed-type terms, $y$, $\OO^2x$ / seed-type terms, $\OO^2y$ & 40.7 / 1.5 & 0.044 / 0.070 & 0.16 / 0.35 \\
Gated network formulation (NN, Halton box, pointwise tests) & 7--8 terms / 7--8 terms & 41.1 / 55.7 & 0.044 / 0.041 & 0.15 / 0.15 \\
HAVOK, measured forcing ($r=15$) & -- & -- & 0.093 / 0.138 & 0.37 / 0.35 \\
HAVOK, forcing frozen ($r=15$) & -- & -- & 0.14 / 0.19 & 0.64 / 0.58 \\
\bottomrule
\end{tabular}
\end{table*}

The delay-coordinate baseline \cite{Brunton2017Havok} uses a
backward-looking Hankel matrix of $x$ and $y$ with 50 delays per channel
spaced 2\,ms apart (a 98\,ms window) and rank $r=15$, and fits
$\dot{\bm v}_{1:r-1}=\bm A\bm v_{1:r-1}+\bm B v_r$ in the delay
coordinates. Fed with the measured forcing coordinate $v_r$ over the
prediction horizon, as in the original HAVOK formulation, it predicts
$x$ ten milliseconds ahead with NRMSE 0.093, about twice the error of the
closure models (0.044); with the forcing held at its value at the start
of the horizon, which makes it a genuine prediction comparable to
SPIRAL-PO, the error is 0.14 and grows to 0.64 at 50\,ms. Its 210 coefficients carry no mechanical
meaning. The two approaches answer different questions: delay
coordinates give a predictive surrogate; SPIRAL-PO gives a small set of
labeled terms whose presence can be tested.

Replacing the amplitude-stratified folds of the gated network formulation
by contiguous time blocks (with or without forward chaining) reduces its
selection from seven terms per equation to the cross-channel term (and
$\OO^2y$ in the $y$-equation) at essentially the same held-out gain
(40.2 / 55.4\% instead of 41.1 / 55.7\%), showing that interleaved folds
let correlated neighboring samples leak between training and validation.

\subsection{Structural mismatch}
\label{ssec:mismatch}

Table~\ref{tab:mismatch} reports the weak-form selection when the
data-generating model differs from the assumed structure. With
anisotropic bearings (diagonal stiffness $1.3k_{xx}$ and $0.7k_{yy}$) and
with all physical parameters perturbed by up to $\pm20\%$, the
physical terms are selected as in the nominal case, with $\dot\OO y$ as
an additional term in the $y$-equation. With a Bouc--Wen hysteretic
element at bearing~1, a history-dependent force outside both the library
and the closure class, the method does not flag the missing mechanism:
it absorbs it into additional cubic forms ($x^3$, $x^2y$), and the
residual predictability check of Step~7 does not distinguish this case
from the nominal one (held-out $R^2$ of the residual network 0.05 and
0.07 in both). Under library misspecification the method therefore
misassigns the missing mechanism rather than stopping early; the
admission of competing forms and the physical plausibility of the
selected terms are the indications to watch. In the run-up case (0 to
40\,Hz) the weak form selects the physical terms plus $\OO^2x$,
$\dot\OO x$ and $\OO^2y$ and misses the gyroscopic term of the
$y$-equation, and the gated loops are unreliable (the valid-statistics
loop admits nothing and its 50-ms predictions diverge). This case lies
outside the present range of the method.

\begin{table}[!t]
\caption{Weak-form selection under structural mismatch (noise seed 31).}
\label{tab:mismatch}
\centering
\tablefont
\begin{tabular}{p{1.70cm} p{2.70cm} p{1.85cm}}
\toprule
Case & $x$-equation & $y$-equation \\
\midrule
Nominal & $\OO\dot y$, $y$, $(x^2{+}y^2)x$ & $\OO\dot x$, $x$, $(x^2{+}y^2)y$ \\
Anisotropic & $\OO\dot y$, $y$, $(x^2{+}y^2)x$ & $\OO\dot x$, $x$, $(x^2{+}y^2)y$, $\dot\OO y$ \\
Parameters $\pm20\%$ & $\OO\dot y$, $y$, $(x^2{+}y^2)x$ & $\OO\dot x$, $x$, $(x^2{+}y^2)y$, $\dot\OO y$ \\
Bouc--Wen hysteresis & $\OO\dot y$, $y$, $(x^2{+}y^2)x$, $x^3$, $x^2y$ & $\OO\dot x$, $x$, $(x^2{+}y^2)y$ \\
Run-up & $\OO\dot y$, $y$, $(x^2{+}y^2)x$, $\OO^2x$, $\dot\OO x$ & $x$, $(x^2{+}y^2)y$, $\OO^2y$ \\
\bottomrule
\end{tabular}
\end{table}

\subsection{How much of the closure does the library capture?}
\label{ssec:capacity}

Because the simulation provides the noise-free trajectory, the
remainder $\epsilon_0$ of \eqref{eq:closure_class} can be estimated
directly: we regress the noise-free post-seed acceleration on the full
library along the trajectory and record the share of its energy that is
explained (Table~\ref{tab:capacity}). Even with all eight candidates and
no noise, the library captures only 63\% ($x$) and 76\% ($y$) of the
closure in the nominal case, i.e.\ $\epsilon_0\approx0.61$ and $0.49$.
The instantaneous closure assumption therefore holds only loosely for
this rotor: much of the hidden-state influence depends on the tilt
history, which no function of the instantaneous $(x,y,\dot x,\dot
y,\OO)$ can represent. Consistently, after the weak-form model is
fitted, the residual network still predicts about 5--7\% of the
held-out residual, at the edge of significance (one-sided $p=0.04$--$0.06$).
Without the cubic stiffness the library captures 95--96\%, and in the
run-up case 98\%.

\begin{table}[!t]
\caption{Share of the noise-free post-seed acceleration energy captured
by the full library (fitted on the trajectory), and the corresponding
$\epsilon_0=\sqrt{1-R^2}$.}
\label{tab:capacity}
\centering
\tablefont
\begin{tabular}{l cc cc}
\toprule
 & \multicolumn{2}{c}{$x$-equation} & \multicolumn{2}{c}{$y$-equation} \\
Case & $R^2$ & $\epsilon_0$ & $R^2$ & $\epsilon_0$ \\
\midrule
Nominal & 0.63 & 0.61 & 0.76 & 0.49 \\
No cubic stiffness & 0.96 & 0.21 & 0.95 & 0.22 \\
No gyroscopic & 0.61 & 0.63 & 0.76 & 0.49 \\
Anisotropic & 0.68 & 0.57 & 0.86 & 0.38 \\
Bouc--Wen hysteresis & 0.47 & 0.73 & 0.72 & 0.53 \\
Parameters $\pm20\%$ & 0.56 & 0.67 & 0.76 & 0.49 \\
Run-up & 0.98 & 0.12 & 0.98 & 0.13 \\
\bottomrule
\end{tabular}
\end{table}

\subsection{Model-based reconstruction of the hidden tilt}
\label{ssec:ekf_reconstruction}

The identified closure is written in observable coordinates and contains
no tilt variable, so it cannot by itself return $\alpha(t)$, $\beta(t)$.
If the structure of the hidden dynamics is \emph{known}, a state
estimator can reconstruct them; this is a separate, model-based step,
not an output of SPIRAL-PO. We use an extended Kalman filter on the full
8-state model $\bm{z} = \{x,y,\alpha,\beta,\dot x,\dot y,\dot\alpha,
\dot\beta\}$ with \eqref{eq:full_rotor}--\eqref{eq:KL_Cb_rotor} as the
process model and $\{x,y\}$ as the only measurement. The mean is
propagated by fourth-order Runge--Kutta integration at $dt=10^{-3}$\,s
and the covariance with the transition matrix of the linearized model,
\begin{equation}
  \bm{\Psi}_k = \exp\bigl(dt\,\bm{F}_k\bigr),
  \label{eq:ekf_Phi}
\end{equation}
where $\bm{F}_k=\partial\bm{f}/\partial\bm{z}$ is the state Jacobian,
\begin{equation}
\begin{split}
  \bm{F}_k &=
  \begin{bmatrix}
    \bm{0}_4 & \bm{I}_4 \\
    \bm{A}_k & \bm{B}_k
  \end{bmatrix},\\
  \bm{A}_k &= -\bm{M}^{-1}\bigl[\bm{K}_L+\bm{J}_{NL}(\bq)\bigr],\\
  \bm{B}_k &= -\bm{M}^{-1}\bigl[\bm{C}_b+\OO\bm{G}\bigr],
\end{split}
  \label{eq:ekf_F}
\end{equation}
with the Jacobian of the cubic bearing force
\begin{equation}
\begin{split}
  \bm{J}_{NL}(\bq)=&\frac{\partial[\bm{K}_{NL}(\bq)\bq]}{\partial\bq}\\
  =&\sum_{i=1}^{2} k_{NLi}\Bigl[\norm{\bm{L}_i\bq}^2\,\bm{L}_i^{\mathsf T}\bm{L}_i
  \\& +2\,(\bm{L}_i^{\mathsf T}\bm{L}_i\bq)(\bm{L}_i^{\mathsf T}\bm{L}_i\bq)^{\mathsf T}\Bigr],
  \label{eq:ekf_JNL}
  \end{split}
\end{equation}
evaluated by finite differences in the implementation. The measurement
matrix is $\bm{H} = [\bm{I}_2\;\bm{0}_{2\times6}]$; the remaining
settings are listed in Table~\ref{tab:ekf_params}. The linearized system
is observable from $\{x,y\}$ (observability matrix of rank 8 at a
representative speed), although its singular values span about twelve
orders of magnitude, so the tilt states are weakly conditioned.

\begin{table*}[ht]
\caption{Extended Kalman filter settings for the 8-state system
$\bm{z}=\{x,y,\alpha,\beta,\dot x,\dot y,\dot\alpha,\dot\beta\}$. State
scales $\bm{s}=\{\sigma_x,\sigma_y,\sigma_x,\sigma_y,\OO_{\max}\sigma_x,
\OO_{\max}\sigma_y,\OO_{\max}\sigma_x,\OO_{\max}\sigma_y\}$, with $\sigma_x$, $\sigma_y$ the standard deviations of the measured channels (tilts per unit
span) and $\OO_{\max}$ the highest spin speed.}
\label{tab:ekf_params}
\centering
\tablefont
\begin{tabular}{p{1.40cm} p{4.70cm} p{8.50cm}}
\toprule
Quantity & Description & Value \\
\midrule
$dt$       & integration / update step
           & $10^{-3}\,$s \\
$\bm{Q}$   & process-noise covariance, scaled to the response
           & $q_0\,\diag(\bm{s}^2)$,\ $q_0=10^{-4}$ \\
$\bm{R}$   & measurement-noise covariance, auto-estimated
           & $\diag(\sigma_x^2,\sigma_y^2)$,\ $\hat\sigma^2=\mathrm{Var}(\Delta^2 q)/6$ \\
$\bm{P}_0$ & initial covariance (positions, tilts, velocities, tilt rates)
           & $\diag(\bm{s}^2)$ \\
$\bm{z}_0$ & initial state ($x,y$ from first sample, rest zero)
           & $\{x(t_1),y(t_1),0,0,0,0,0,0\}^{\mathsf T}$ \\
\bottomrule
\end{tabular}
\end{table*}

On the independent synthetic record (noise seed 77, unbalance phase
$+0.7$\,rad), with the filter receiving only noisy $x,y$, the tilt
angles are reconstructed with relative RMSE 1.3\% ($\alpha$) and 1.2\%
($\beta$) and correlation 0.9999 (Table~\ref{tab:state_errors},
Fig.~\ref{fig:synthetic_ab}); the error is nearly uniform over the
record, including the resonance crossing. This result depends on the assumed model:
the filter is given the correct equations of the hidden dynamics. If the
structure of the hidden subsystem were unknown, SPIRAL-PO's observable
equations alone could not return labeled tilt trajectories; a
data-driven surrogate observer could predict $x,y$ but would not return
physical states.

\begin{figure*}[ht]
    \centering
    \includegraphics[width=1\textwidth]{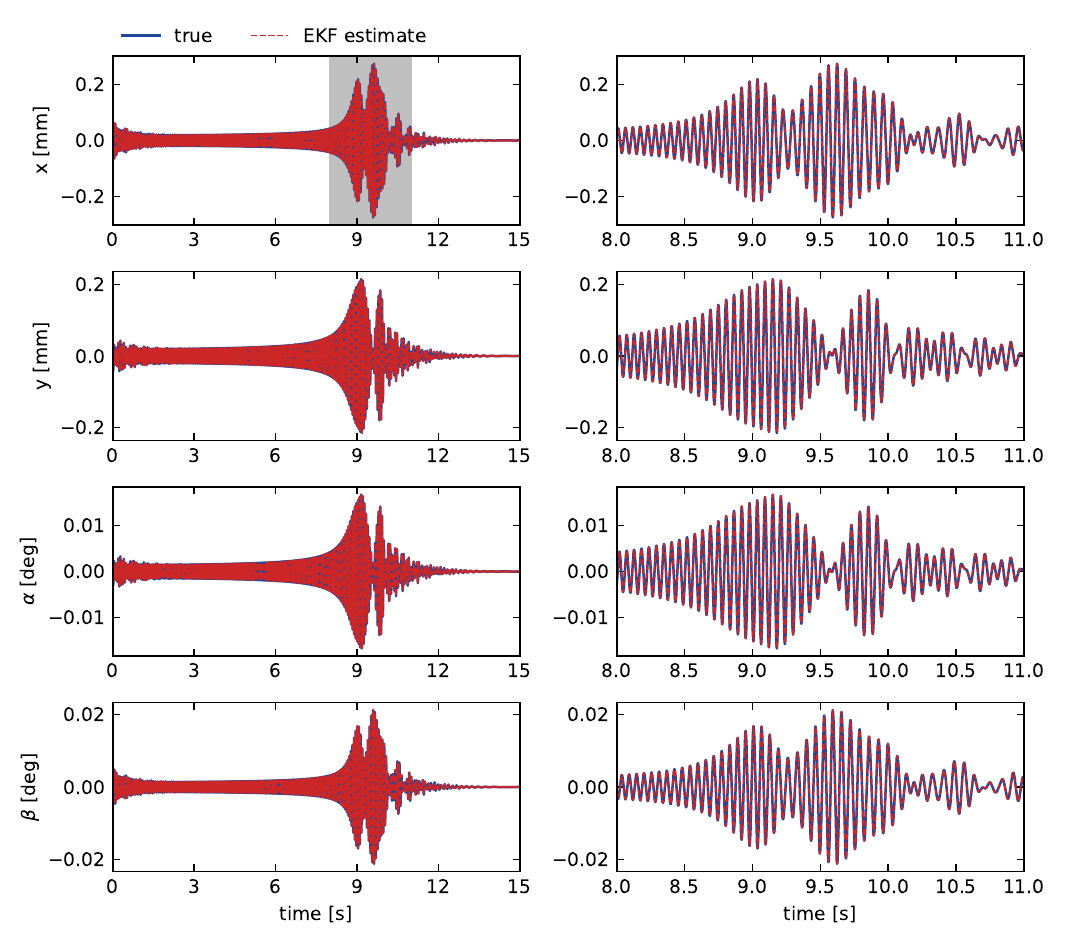}
    \caption{Model-based reconstruction of the hidden tilt with an
    extended Kalman filter that is given the 4-DOF model and only the
    noisy $x,y$ of an independent synthetic record (15\,s coast-down,
    40 to 0\,Hz). Solid: true; dashed: estimate. Tilt angles in degrees. Right column: zoom
    from 8 to 11\,s (shaded in the top left panel).}
    \label{fig:synthetic_ab}
\end{figure*}

\begin{table}[ht]
\caption{Model-based state reconstruction on the independent synthetic
record (Fig.~\ref{fig:synthetic_ab}).}
\label{tab:state_errors}
\centering
\tablefont
\setlength{\tabcolsep}{4pt}
\begin{tabular}{l c c c}
\toprule
State & RMSE & rel.\ error & correlation \\
\midrule
$x$ & $5.4\times10^{-7}$\,m & 1.2\% & 0.9999 \\
$y$ & $4.6\times10^{-7}$\,m & 1.2\% & 0.9999 \\
$\alpha$ (hidden) & $6.7\times10^{-7}$\,rad & 1.3\% & 0.9999 \\
$\beta$ (hidden) & $7.3\times10^{-7}$\,rad & 1.2\% & 0.9999 \\
\bottomrule
\end{tabular}
\end{table}

\subsection{Robustness of the reconstruction: parameter mismatch and
noise}
\label{ssec:ekf_robustness}

Because the filter of Sect.~\ref{ssec:ekf_reconstruction} shares the
structure and parameters of the simulator, its accuracy could reflect a
self-consistent model rather than the information content of $x,y$. Two
controls address this: the bearing parameters of the filter are
perturbed jointly (all entries of $\bm{K}_b$, $k_{NL}$ and $\bm{C}_b$ multiplied by $1+\delta$), and
the measurement noise is raised above the nominal 2\%
(Tables~\ref{tab:ekf_mismatch} and~\ref{tab:ekf_noise}; the filter is
run on an independent record in every case). The simulated truth is the
nominal benchmark of Table~\ref{tab:sim_params} throughout; only the
filter's bearing parameters and the noise level are varied.

\begin{table*}[ht]
\caption{Hidden-state reconstruction under bearing-parameter mismatch
(measurement noise held at the nominal $2\%$). Mismatch $\delta$ is
applied jointly to $k_{L}$, $k_{NL}$, $c_{b}$ in the filter model; the
$\delta=0\%$ row is the matched baseline. Metrics are
computed against the true hidden trajectories on an independent
validation record; RMSE in rad ($\alpha$, $\beta$) and m ($x,y$).}
\label{tab:ekf_mismatch}
\centering
\tablefont
\begin{tabular}{l l l l l l}
\toprule
Mismatch $\delta$
  & $\alpha$ RMSE & $\alpha$ corr.
  & $\beta$ RMSE & $\beta$ corr.
  & $x,y$ RMSE \\
\midrule
$0\%$ (matched) & $6.7\times10^{-7}$ & $0.9999$ & $7.2\times10^{-7}$ & $0.9999$ & $5.0\times10^{-7}$ \\
$5\%$           & $6.8\times10^{-7}$ & $0.9999$ & $7.4\times10^{-7}$ & $0.9999$ & $5.1\times10^{-7}$ \\
$10\%$          & $7.4\times10^{-7}$ & $0.9999$ & $8.1\times10^{-7}$ & $0.9999$ & $5.5\times10^{-7}$ \\
$20\%$          & $9.2\times10^{-7}$ & $0.9999$ & $1.0\times10^{-6}$ & $0.9999$ & $6.9\times10^{-7}$ \\
\bottomrule
\end{tabular}
\end{table*}

\begin{table*}[ht]
\caption{Hidden-state reconstruction under increasing measurement
noise (bearing parameters matched, $\delta=0$). Noise is expressed as
a percentage of each channel's standard deviation; the $2\%$ row is the
nominal baseline. ``Stable window'' is the longest leading segment over
which the filter remains convergent (the record is 15\,s long).}
\label{tab:ekf_noise}
\centering
\tablefont
\begin{tabular}{l l l l l l}
\toprule
Noise $\sigma$
  & $\alpha$ RMSE & $\alpha$ corr.
  & $\beta$ RMSE & $\beta$ corr.
  & Stable window \\
\midrule
$2\%$ (nominal) & $6.7\times10^{-7}$ & $0.9999$ & $7.2\times10^{-7}$ & $0.9999$ & $15\,$s \\
$5\%$           & $1.2\times10^{-6}$ & $0.9997$ & $1.4\times10^{-6}$ & $0.9997$ & $15\,$s \\
$10\%$          & $1.8\times10^{-6}$ & $0.9994$ & $2.1\times10^{-6}$ & $0.9994$ & $15\,$s \\
$20\%$          & $2.7\times10^{-6}$ & $0.9986$ & $3.2\times10^{-6}$ & $0.9985$ & $15\,$s \\
\bottomrule
\end{tabular}
\end{table*}

The degradation is gradual. A 20\% joint mismatch of the bearing
parameters increases the tilt RMSE from about $7\times10^{-7}$ to
$1.0\times10^{-6}$\,rad with correlation 0.9999, and noise up to 20\%
of the channel standard deviation keeps the tilt correlation above
0.998, with the filter convergent over the whole record.
These sweeps perturb parameters, not structure: they show that the
reconstruction tolerates moderate parameter error, not that it would
survive an incorrect model of the hidden subsystem.

\section{Experimental test rig}
\label{sec:testrig}

The experimental test rig used to validate the SPIRAL-PO framework
was originally developed to investigate physics-based, data-driven
identification and active vibration control of a laboratory-scale
vertical-shaft rotating machine \cite{Piramoon2024}, designed to
replicate the key dynamical features of industrial centrifugal
machinery in a controllable, instrumented setting.

\subsection{Laboratory-scale rotating system}
\label{subsec:test_rig}

Experimental validation was performed using a laboratory-scale vertical rotor-bearing platform specifically configured to acquire high-quality vibration measurements for nonlinear system identification. The test rig enables repeatable experiments under constant-speed, acceleration, and deceleration operating conditions, thereby providing a representative dataset for evaluating the proposed SPIRAL-PO framework. An overview of the experimental apparatus is presented in Fig.~\ref{fig:test_rig}, while the corresponding four-degree-of-freedom mathematical model employed throughout this paper is illustrated in Fig.~\ref{fig:4DOF_rotor}.

The mechanical drive system comprises a variable-speed AC motor \textcircled{1} mounted on a rigid support frame \textcircled{2}. The motor drives the rotor assembly through a flexible shaft \textcircled{4}, which carries the rotor with mounted test tubes \textcircled{5}. The shaft is supported by the upper bearing housing \textcircled{3}, the primary rotor-bearing assembly \textcircled{6}, and the lower bearing housing \textcircled{7}. This configuration provides the lateral compliance necessary to excite coupled translational and rotational dynamics over the operating speed range considered in this study. For operator protection, the rotating assembly is enclosed by a transparent safety shield \textcircled{8}.

Only lateral translation of the disk is measured; the disk tilt is not, and the spin speed was not recorded (Sect.~\ref{subsec:data_acquisition}). The rig records are used to examine the procedure on real data and to show how strongly its result depends on the speed signal; they do not provide an independent confirmation of the mechanisms identified on the synthetic benchmark.
\begin{figure}[!t]
    \centering
    \includegraphics[width=.65\linewidth]{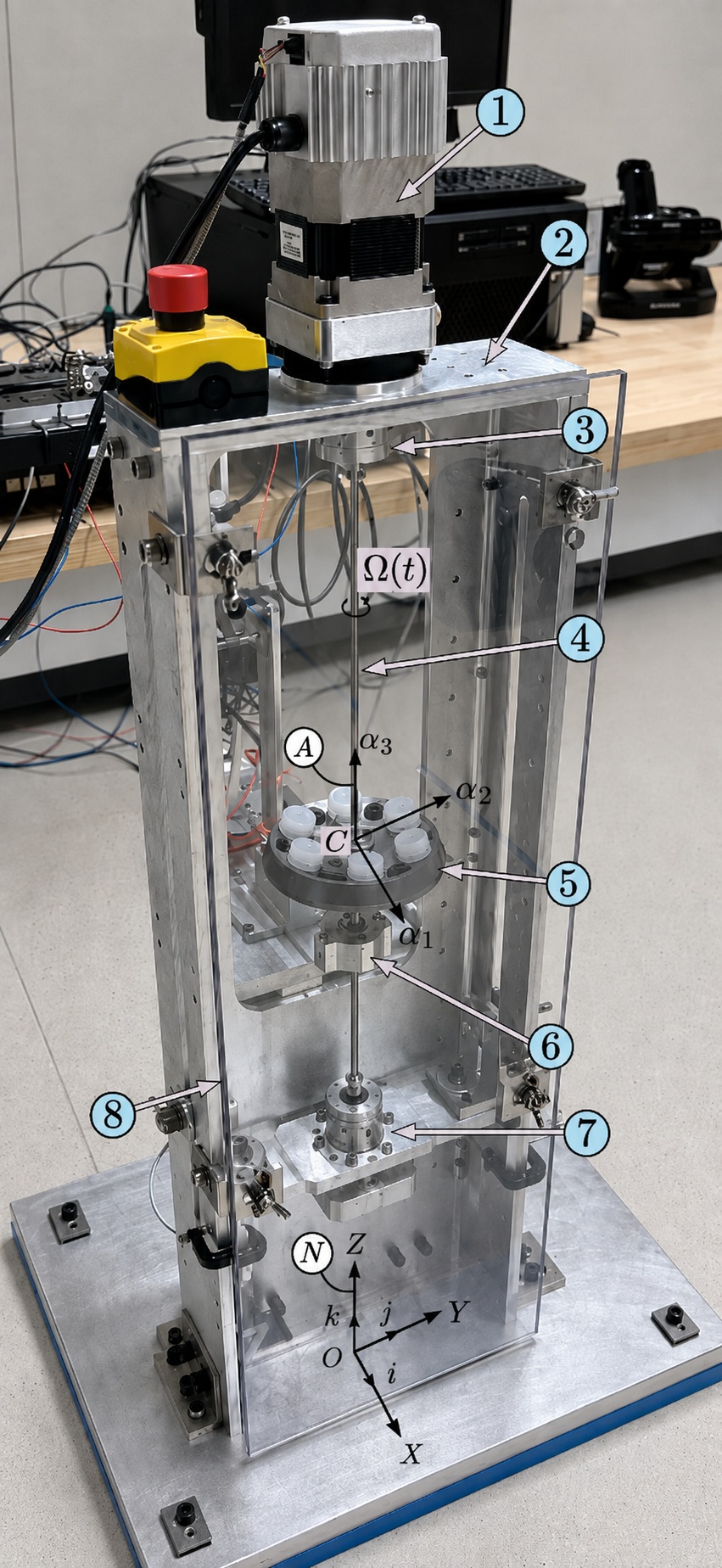}
\caption{Laboratory-scale vertical rotor test rig and its components: an AC
motor \textcircled{1} is mounted on the chassis \textcircled{2} and applies
torque to the drive mechanism, which includes the disk rotor \textcircled{5}
mounted on a slender, flexible shaft \textcircled{4}. This shaft extends
through the rotor and connects to the actuator force transmission assembly,
which in turn connects to the rotor bearing \textcircled{6}. This assembly is
linked to the bearing housing at the top \textcircled{3} and at the bottom
\textcircled{7}. A safety shield \textcircled{8} covers the front panel.}
\label{fig:test_rig}
\end{figure}

The mechanical assembly consists of a disk-shaped rotor mounted on a
slender flexible steel shaft. Multiple circumferential apertures in the
disk permit calibrated trial masses at prescribed radial and angular
locations, so imbalance magnitude and phase can be varied in a
controlled, repeatable manner; the resulting rotating imbalance force
provides the persistent synchronous excitation used for identification.

The rotor is driven by a variable-frequency AC electric motor through a
gearbox, damping plate, and flexible coupling, transmitting torque
through a $6~\mathrm{mm}$-diameter steel shaft. Lovejoy\textsuperscript{\textregistered}
couplings accommodate small shaft misalignments, attenuate impulsive
loads, and limit the transmission of motor-side disturbances to the
rotor. A centrifugal clutch between the drive shaft and rotor further
reduces the influence of small geometric and assembly misalignments on
the measured response, and the bearing housing is supported by a
cylindrical elastomeric element providing structural flexibility and
damping in both lateral directions.

The complete system is mounted on an aluminum base plate and C-channel
chassis, which provides repeatable mounting locations for the
displacement sensors, accelerometers, actuators, and safety shield.
The two lateral actuators are part of the general-purpose platform but
are not used to apply feedback forces during the open-loop
system-identification experiments considered here.

The motor command is prescribed as a constant-speed, acceleration, or
deceleration profile. For a linear speed ramp,
\begin{equation}
    \Omega(t)=\Omega_{0}+\dot{\Omega} t,
    \label{eq:speed_profile}
\end{equation}
where $\Omega_{0}$ is the initial angular velocity and $\dot{\Omega}$ is the constant angular acceleration (positive: acceleration; negative:
deceleration).
Variable-speed experiments excite the rotor over a continuous frequency
interval and pass the resonance; the records analyzed here are
coast-downs; the deceleration is not constant
(Sect.~\ref{subsec:data_acquisition}).

\subsection{Measurement and data}
\label{subsec:data_acquisition}

Lateral displacement of the disk is measured by two Micro-Epsilon
optoNCDT~1420 non-contact laser triangulation sensors mounted at the
disk in the $X$ and $Y$ directions, which add no mass, stiffness or
damping to the rotor. Two kinds of record are used.

\emph{Simultaneous record.} Both sensors were recorded simultaneously
through a Quanser Q8 data-acquisition board (16-bit, MATLAB/Simulink
with QUARC) at 1\,kHz during a coast-down from about 32\,Hz to rest in
6.7\,s. This record supports the bivariate library of
Table~\ref{tab:CandidateLibrary}. The two channels are strongly coherent
at the spin frequency (magnitude-squared coherence up to 0.99), as a
simultaneous measurement of one orbit should be.

\emph{Single-sensor coast-downs.} In a separate session, one sensor
(4\,kHz, RS422) recorded coast-downs from 60\,Hz in the $x$ and $y$
directions, one direction per run: the sensor was moved between runs, so
these records ($X$ and $Y$) are not simultaneous. A third run measured
the axial direction and is not used, since the model is planar. They are used for a single-channel repeatability check
(Sect.~\ref{ssec:rig_single}).

\emph{Spin speed.} The speed was not recorded by a tachometer or
encoder. The speed curve supplied with the simultaneous record was
inferred from the response frequency under the assumption of a
constant deceleration, which gives a straight line from 32.4\,Hz to
zero. We estimate the speed instead by order tracking of the response
itself: the ridge of the short-time Fourier transform (1.95\,Hz bins) is
followed from the start of the coast-down downwards, with the stationary
background removed, and a smooth curve is obtained by linear
interpolation through the center of each constant-frequency step; for
the simultaneous record the curves from the two channels are averaged.
Fig.~\ref{fig:rig_sim} compares the two speed curves. The response
follows the tracked curve: the spectral power along it is 55--57\,dB
above the background, against 32--43\,dB along the straight line, and
the second harmonic of the unbalance response is visible along twice the
tracked speed (22--23\,dB above background) but not along twice the
straight line ($\le4$\,dB). The straight line lies above the tracked
speed by 4.1\,Hz RMS and up to 5.9\,Hz (for example 20.3 versus
15.8\,Hz at 2.5\,s), because the deceleration is not constant: it falls
from about 6.5\,Hz/s at high speed to about 3\,Hz/s near rest. The
unbalance phase $\vartheta(t)=\int\OO\,\mathrm{d}t$ is formed from the
tracked speed; its absolute value is unknown, so the forcing amplitude
and phase are estimated as part of the seed. Both the tracked speed and
the phase carry an uncertainty of the order of the bin width (about
$\pm1$\,Hz).

\begin{figure*}[ht]
\centering
\includegraphics[width=\textwidth]{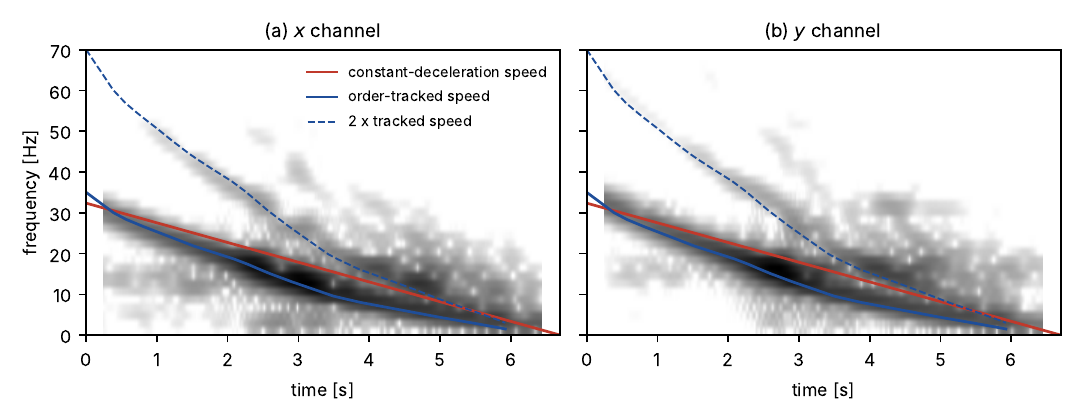}
\caption{Simultaneous coast-down record of the rig: spectrograms of the
$x$ (a) and $y$ (b) channels (1\,kHz, 0.512\,s windows) with the
constant-deceleration speed curve supplied with the record (red), the
order-tracked speed (blue) and twice the tracked speed (dashed). The
unbalance response and its second harmonic follow the tracked speed.}
\label{fig:rig_sim}
\end{figure*}

\subsection{Bivariate identification on the simultaneous record}
\label{ssec:rig_bivariate}

The procedure is applied to the simultaneous record exactly as to the
synthetic data: seed \eqref{eq:seed_rotor}, the eight candidates
of Table~\ref{tab:CandidateLibrary}, weak form with $h=25$\,ms, and the
Bonferroni threshold $|t|>2.73$. The part of the record with tracked
speed above 3\,Hz (5.45\,s, 106 test functions per equation) is used;
Savitzky--Golay windows are chosen from $\{11,21,31\}$ samples at 1\,kHz.
Because the result depends on a speed estimate, the identification is
repeated for nine tracking settings (windows of 256, 512 and 1024
samples; speed from the $x$ channel, the $y$ channel, or both).

\begin{table*}[!t]
\caption{Bivariate weak-form identification on the simultaneous rig
record. $t$: weak-form $t$ value with the tracked speed (512-sample
window, both channels); bold: admitted ($|t|>2.73$). Count: number of
the nine tracking settings in which the term is admitted. Line: admitted
when the constant-deceleration speed curve is used instead. Pointwise
OLS with classical standard errors admits all eight candidates in both
equations.}
\label{tab:rig_bivariate}
\centering
\tablefont
\begin{tabular}{l r c c @{\hspace{2em}} l r c c}
\toprule
\multicolumn{4}{c}{$x$-equation} & \multicolumn{4}{c}{$y$-equation} \\
\cmidrule(lr){1-4}\cmidrule(lr){5-8}
Term & $t$ & count & line & Term & $t$ & count & line \\
\midrule
$\OO\dot y$ & $-2.7$ & 2/9 & yes & $\OO\dot x$ & $1.8$ & 3/9 & -- \\
$y$ & $-1.5$ & 0/9 & -- & $x$ & $-2.1$ & 4/9 & yes \\
$(x^2{+}y^2)x$ & $2.2$ & 0/9 & yes & $(x^2{+}y^2)y$ & $1.4$ & 1/9 & yes \\
$(x^2{+}y^2)\dot x$ & $1.5$ & 0/9 & -- & $(x^2{+}y^2)\dot y$ & $-1.2$ & 3/9 & -- \\
$x^3$ & $-0.6$ & 0/9 & -- & $y^3$ & $-2.0$ & 3/9 & yes \\
$x^2y$ & $1.6$ & 0/9 & -- & $xy^2$ & $0.6$ & 1/9 & -- \\
$\OO^2x$ & $\mathbf{-12.1}$ & 9/9 & yes & $\OO^2y$ & $\mathbf{-6.5}$ & 9/9 & yes \\
$\dot\OO x$ & $0.4$ & 0/9 & -- & $\dot\OO y$ & $0.3$ & 4/9 & -- \\
\bottomrule
\end{tabular}
\end{table*}

Table~\ref{tab:rig_bivariate} gives the result. Only the speed-dependent
stiffness terms $\OO^2x$ and $\OO^2y$ are admitted, with negative
coefficients (stiffening that grows with speed), and they are admitted
under all nine tracking settings. The gyroscopic terms have the sign
pattern of the benchmark and of the skew-symmetric gyroscopic matrix
(negative in the $x$-equation, positive in the $y$-equation; compare
Table~\ref{tab:coefficients}), but they stay below the threshold with the
reference setting and are admitted in only two and three of the nine
settings. The isotropic cubic term is not detected, and its coefficient
has the sign of softening rather than hardening. Pointwise OLS with
classical standard errors would admit all eight candidates.

With the constant-deceleration speed curve instead of the tracked one,
the same procedure admits the gyroscopic and cubic terms in the
$x$-equation and four terms in the $y$-equation (column ``line''). The synthetic
benchmark shows that a speed error of this size produces such
selections. When the identification is given a straight speed line that
starts 4\,Hz above the true speed and ends at the true final speed, it
admits $\OO\dot y$ and $\OO\dot x$ in the gyroscopic-free control case,
admits $(x^2{+}y^2)\dot x$, $\OO^2x$ and $\OO^2y$ in the cubic-free
control, and loses the correct cubic terms in the nominal case; with a
speed tracked at the rig's 1.95\,Hz resolution it returns the correct
terms in all three cases. An error of a few hertz in the
speed shifts the unbalance phase by several revolutions within seconds,
so the seed no longer represents the forcing and the unexplained
synchronous response is absorbed by whichever candidates oscillate at
the spin frequency. The selections obtained with the constant-deceleration
curve are therefore not evidence for the corresponding mechanisms.

The tracked-speed result has the following interpretation. The $\OO^2x$, $\OO^2y$
terms describe a stiffening that grows with speed; in this library it
can stand for amplitude-dependent stiffening (amplitude and speed are
tied along a coast-down), for a speed-dependent stiffening of the whirl
mode, or for a residual error of the estimated speed near the resonance,
and the present data cannot separate these. The gyroscopic signature has
the expected sign pattern but is not detected reliably in a 5.5\,s record
with a speed known only to about $\pm1$\,Hz. A once-per-revolution
reference, which fixes both the speed and the unbalance phase, and
longer or repeated simultaneous coast-downs are needed to test the
gyroscopic and cubic mechanisms on the rig.

\subsection{Single-channel repeatability}
\label{ssec:rig_single}

The single-sensor coast-downs (60 to 3\,Hz, 4\,kHz, one direction per
run) allow a repeatability check with the terms of the measured channel
$s$ only,
\begin{equation}
  \{\,s^3,\ s^2\dot s,\ \OO^2 s,\ \dot\OO s\,\},
  \label{eq:lib_single}
\end{equation}
tested in the weak form with the Bonferroni threshold for $K=4$
($|t|>2.50$), with the speed tracked as above (2048-sample windows at
4\,kHz, again 1.95\,Hz bins). Each record is low-pass filtered and
decimated to 2\,kHz and the interval from 1 to 24\,s is used; two samples
of record $Y$ carry the sensor's out-of-range code and are interpolated.
These records also contain a stationary 16.6\,Hz component that does not
follow the spin speed; it lies at an undamped natural frequency of the
rig. Applied to
the $x$ channel of the synthetic benchmark, this single-channel analysis
selects $s^3$ (negative) when the cubic stiffness is present and nothing
when it is removed, with the exact or the tracked speed; $\OO^2 s$ has a
negative $t$ value of about $-1$ to $-2.4$ in all cases but is not
admitted.

\begin{table}[!t]
\caption{Single-channel weak-form identification on the single-sensor
coast-downs (seed plus four candidates; Bonferroni $|t|>2.50$), with the
pointwise OLS $t$ value for comparison. Held-out gain: share of the
post-seed residual of another record explained by the selected terms,
with the physical coefficients fixed and only the unbalance amplitude
and phase refitted.}
\label{tab:rig_single}
\centering
\tablefont
\setlength{\tabcolsep}{4pt}
\begin{tabular}{p{1.9cm} rr rr}
\toprule
 & \multicolumn{2}{c}{fit on record $X$} & \multicolumn{2}{c}{fit on record $Y$} \\
\cmidrule(lr){2-3}\cmidrule(lr){4-5}
Term & $t$ weak & $t$ OLS & $t$ weak & $t$ OLS \\
\midrule
$s^3$ & $-0.08$ & $7.93$ & $-2.14$ & $-8.55$ \\
$s^2\dot s$ & $0.70$ & $2.03$ & $0.03$ & $-2.79$ \\
$\OO^2 s$ & $\mathbf{-9.07}$ & $-70.8$ & $\mathbf{-5.59}$ & $-97.6$ \\
$\dot\OO s$ & $\mathbf{3.48}$ & $5.32$ & $\mathbf{2.63}$ & $1.49$ \\
\midrule
Seed natural frequency & \multicolumn{2}{c}{8.9\,Hz} & \multicolumn{2}{c}{10.3\,Hz} \\
Test functions & \multicolumn{2}{c}{432} & \multicolumn{2}{c}{447} \\
Held-out gain, other run & \multicolumn{2}{c}{14.2\% (on $Y$)} & \multicolumn{2}{c}{4.5\% (on $X$)} \\
\bottomrule
\end{tabular}
\end{table}

On both runs the weak form selects $\OO^2 s$ with a negative coefficient
and $\dot\OO s$, and not $s^3$ (Table~\ref{tab:rig_single}); pointwise
OLS would select $s^3$ on both runs, with opposite signs. The terms
identified on one run explain 14\% (fit on $X$, tested on $Y$) and 4.5\%
(fit on $Y$, tested on $X$) of the post-seed residual of the other run.
The speed-dependent stiffening found on the simultaneous record is thus
reproduced on independent coast-downs at higher speed; the $\dot\OO s$
coefficient depends on the derivative of the estimated speed and is not
interpreted.

Without an independent measurement of the disk tilt and of the spin
speed, the model-based state reconstruction of
Sect.~\ref{ssec:ekf_reconstruction} is not applied to the rig.

\section{Discussion}
\label{sec:discussion}

\subsection{Comparison with related methods}
\label{sec:novelty}

Full-state sparse regression, SINDy~\cite{Brunton2016} and its weak form
WSINDy~\cite{Messenger2021}, regresses the dynamics on a library of the
measured state; with $\alpha,\beta$ unmeasured, the library for the tilt
equations cannot be formed. Applied to the observable equations only,
with a library of $x,y$ monomials, these methods address the same
regression problem as SPIRAL-PO, which uses the weak formulation of~\cite{Messenger2021} for its test. The
differences are the physics seed, the mechanism-labeled library derived
from the hidden-coupling terms, the statistical admission rule with
dependence-aware uncertainty in place of a sparsity threshold, and the
knock-out controls; the ablation of Table~\ref{tab:ablation} shows that
the sparsity threshold alone (STLSQ) admits four to seven terms per
equation on the same data. Physics-informed neural
networks~\cite{Raissi2019} require the form of the governing equations
and estimate parameters within it. Delay-coordinate methods based on
Takens' theorem~\cite{Takens1981}, such as HAVOK~\cite{Brunton2017Havok}
and its deep-learning variants~\cite{Bakarji2023}, operate on partial
measurements directly; we compared HAVOK quantitatively in
Sect.~\ref{ssec:ablation}.

\begin{table*}[ht]
\caption{Capabilities in the partial-observation setting of this paper,
as demonstrated here. SINDy/WSINDy are listed in their full-state
formulation; when restricted to the observable equations they apply,
without a physics seed or mechanism-labeled library.}
\label{tab:comparison}
\centering
\tablefont
\begin{tabular}{p{5.21cm} p{3.17cm} p{2.79cm} p{2.98cm}}
\toprule
Property
  & SPIRAL-PO
  & HAVOK \cite{Brunton2017Havok}
  & SINDy / WSINDy \cite{Brunton2016,Messenger2021} \\
\midrule
Uses only the measured coordinates & Yes & Yes & No (full state) \\
Model in physical coordinates & Yes & No (delay coordinates) & Yes \\
Mechanism-labeled terms & Yes & No & Partly \\
Returns labeled hidden states & No (requires an assumed model, Sect.~\ref{ssec:ekf_reconstruction}) & No & n/a \\
Physics seed & Yes & No & No \\
Coefficient uncertainty & Yes (weak form, block bootstrap) & No & In later extensions \\
Knock-out controls reported & Yes & No & No \\
10-ms prediction of $x$ (this benchmark) & NRMSE 0.044 & 0.093 (measured forcing), 0.14 (frozen) & n/a \\
\bottomrule
\end{tabular}
\end{table*}

The contrast with delay-coordinate methods is one of purpose. HAVOK
needs no physics seed, but in this benchmark it predicts the measured
signal less accurately than the closure models, even with its forcing
coordinate supplied, and its coefficients are not mechanical quantities. SPIRAL-PO returns a small set of terms, each
attached to a mechanism, whose presence can be tested and whose absence
in controls can be checked; the knock-out experiments of
Table~\ref{tab:selection} provide the main evidence that the
admitted terms reflect the data-generating mechanisms rather than the
method.

Grey-box identification of the full model, for example by output-error
estimation or joint state and parameter estimation with an augmented
Kalman filter, is the natural alternative when the structure of the
hidden subsystem is known. It assumes that structure and estimates its
parameters. SPIRAL-PO instead tests, with controlled error rates,
whether each candidate mechanism leaves a signature in the measured
equations, without committing to a complete model of the hidden states.
The two approaches are complementary: the terms admitted by SPIRAL-PO
can be used to choose the structure that a grey-box estimator then
fits.

\subsection{Interpretation and limitations}
\label{sec:discussion_obs}

The results are subject to the following limitations.

\emph{The closure is only partly instantaneous.} In the benchmark the
identified closure explains about 40\% ($x$) and 55\% ($y$) of the
post-seed residual on an independent record, mostly through the
cross-channel term. Even without noise, the library captures 63--76\%
of the closure (Table~\ref{tab:capacity}), because much of the
hidden-state influence depends on the tilt history. The unexplained part
is correlated across test functions, which is why the gyroscopic term
is admitted in 15--25\% of the realizations without gyroscopic coupling
(Table~\ref{tab:selection}). The method identifies which mechanisms
leave a detectable instantaneous signature; it does not provide a
complete observable model of the hidden dynamics.

\emph{Coefficients are not physical constants.} With $\epsilon_0$ of
order 0.5--0.6, the coefficients absorb part of the library remainder
and differ between estimators by more than their standard errors
(Table~\ref{tab:coefficients}). Signs and the presence of a term are
meaningful in this benchmark; magnitudes should not be compared directly
with $I_p$ or $k_{NL}$.

\emph{Library misspecification is absorbed, not flagged.} A mechanism
outside the library, such as hysteresis, is represented by additional
library terms (Table~\ref{tab:mismatch}), and in this benchmark the
residual-predictability check of Step~7 does not flag it; the admission
of competing forms and the physical plausibility of the selection are
the available indications.

\emph{Hidden states require a model.} Labeled tilt trajectories were
recovered only with an estimator given the correct hidden-subsystem
equations (Sect.~\ref{ssec:ekf_reconstruction}).

\emph{The experimental evidence is limited.} On the rig, $x$ and $y$
were measured simultaneously but the spin speed was not; with an
order-tracked speed the procedure admits only a speed-dependent
stiffening, consistently over tracking settings and over independent
single-sensor coast-downs, and this term cannot be attributed to a
specific mechanism (Sects.~\ref{ssec:rig_bivariate}
and~\ref{ssec:rig_single}). The gyroscopic signature has the expected
sign pattern but is not detected reliably. The rig also showed how
sensitive the identification is to the speed signal: a speed curve a
few hertz in error produces gyroscopic and cubic terms that are absent
from the data, as the synthetic controls confirm. The mechanism claims of
this paper therefore rest on the synthetic benchmark. A once-per-revolution
reference, longer or repeated simultaneous coast-downs and an
independent tilt measurement are needed for an experimental test, and
are planned.

\emph{Scope.} The speed $\OO(t)$ was known in the synthetic study and
estimated on the rig, and errors of a few hertz in it change the
result; the run-up case shows that the method can fail
when the seed is inadequate. The approach is not specific to rotors:
the same steps apply when a mechanism-labeled library can be derived
for the coupling of hidden to measured coordinates, as in the
semi-empirical aeroelastic models discussed in the Introduction.

\section{Conclusions}
\label{sec:conclusions}

SPIRAL-PO identifies, from measured coordinates only, which
hidden-state mechanisms leave a detectable signature in the observable
equations of motion. It combines a physics seed, a
candidate library derived from the coupling terms of the hidden
subsystem, a weak-form regression with errors that are independent
across test functions, and a Bonferroni-corrected admission rule, with
block-bootstrap uncertainty, time-blocked validation and a residual
predictability check as diagnostics. On a strongly nonlinear 4-DOF
rotor whose tilt dynamics are simulated and then hidden, it recovers the
gyroscopic, cubic-stiffness and cross-channel signatures from
$x(t),y(t)$ in every noise realization at 2\% noise. In knock-out
controls it never admits an absent cubic stiffness and admits an absent
gyroscopic coupling in 15--25\% of realizations; pointwise regression
with classical standard errors admits the absent cubic term in about
half of the realizations, and a gated network formulation with pointwise
tests admits the absent mechanisms in nearly all. Ablations show that a residual
network and an off-trajectory projection stage do not contribute, and
that the decisive ingredient is valid inference on serially correlated
data.

Three quantitative caveats bound these results: the identified closure
is dominated by the cross-channel term; the instantaneous library
captures only part of the closure ($\epsilon_0\approx0.5$--$0.6$), so
coefficients are not physical constants and the nominal false-admission
level is not reached for the gyroscopic term; and hidden states are recovered only when a
model of the hidden subsystem is assumed. On the laboratory rig, with
simultaneous $x,y$ measurement but no measured speed, the method admits
only a speed-dependent stiffening once the speed is tracked from the
response; the gyroscopic and cubic terms obtained with an assumed
constant-deceleration speed are consistent with the speed error, which
produces such terms in simulation. Testing
the hidden-state mechanisms experimentally requires a once-per-revolution
speed and phase reference.

\backmatter

\section*{Statements and Declarations}

\subsubsection*{Author Contributions}
M.A.A. and S.P. conceived and designed the study. M.A.A. developed the
theoretical framework. M.A.A. and S.P. ran the synthetic simulations.
S.P. collected the experimental data and implemented the numerical algorithms. M.A.A. wrote the majority of the manuscript, and S.P. edited portions of it and both authors contributed to the interpretation of the results. All authors reviewed and approved the final manuscript.

\subsubsection*{Funding}
Not applicable.

\subsubsection*{Conflicts of Interest/Competing Interests}
The authors have no relevant financial or non-financial interests to
disclose.

\subsubsection*{Ethics Approval}
Not applicable.

\subsubsection*{Consent to Participate}
Not applicable.

\subsubsection*{Consent for Publication}
Not applicable.

\subsubsection*{Data Availability}
The MATLAB code, the experimental rig records, and the scripts, settings,
random seeds and candidate-library definitions that generate the synthetic
datasets are openly available on Zenodo at
\url{https://doi.org/10.5281/zenodo.23875336}.

\bibliography{SPIRAL_PO}


\begin{thebibliography}{27}
\ifx \bisbn   \undefined \def \bisbn  #1{ISBN #1}\fi
\ifx \binits  \undefined \def \binits#1{#1}\fi
\ifx \bauthor  \undefined \def \bauthor#1{#1}\fi
\ifx \batitle  \undefined \def \batitle#1{#1}\fi
\ifx \bjtitle  \undefined \def \bjtitle#1{#1}\fi
\ifx \bvolume  \undefined \def \bvolume#1{\textbf{#1}}\fi
\ifx \byear  \undefined \def \byear#1{#1}\fi
\ifx \bissue  \undefined \def \bissue#1{#1}\fi
\ifx \bfpage  \undefined \def \bfpage#1{#1}\fi
\ifx \blpage  \undefined \def \blpage #1{#1}\fi
\ifx \burl  \undefined \def \burl#1{\textsf{#1}}\fi
\ifx \doiurl  \undefined \def \doiurl#1{\url{https://doi.org/#1}}\fi
\ifx \betal  \undefined \def \betal{\textit{et al.}}\fi
\ifx \binstitute  \undefined \def \binstitute#1{#1}\fi
\ifx \binstitutionaled  \undefined \def \binstitutionaled#1{#1}\fi
\ifx \bctitle  \undefined \def \bctitle#1{#1}\fi
\ifx \beditor  \undefined \def \beditor#1{#1}\fi
\ifx \bpublisher  \undefined \def \bpublisher#1{#1}\fi
\ifx \bbtitle  \undefined \def \bbtitle#1{#1}\fi
\ifx \bedition  \undefined \def \bedition#1{#1}\fi
\ifx \bseriesno  \undefined \def \bseriesno#1{#1}\fi
\ifx \blocation  \undefined \def \blocation#1{#1}\fi
\ifx \bsertitle  \undefined \def \bsertitle#1{#1}\fi
\ifx \bsnm \undefined \def \bsnm#1{#1}\fi
\ifx \bsuffix \undefined \def \bsuffix#1{#1}\fi
\ifx \bparticle \undefined \def \bparticle#1{#1}\fi
\ifx \barticle \undefined \def \barticle#1{#1}\fi
\bibcommenthead
\ifx \bconfdate \undefined \def \bconfdate #1{#1}\fi
\ifx \botherref \undefined \def \botherref #1{#1}\fi
\ifx \url \undefined \def \url#1{\textsf{#1}}\fi
\ifx \bchapter \undefined \def \bchapter#1{#1}\fi
\ifx \bbook \undefined \def \bbook#1{#1}\fi
\ifx \bcomment \undefined \def \bcomment#1{#1}\fi
\ifx \oauthor \undefined \def \oauthor#1{#1}\fi
\ifx \citeauthoryear \undefined \def \citeauthoryear#1{#1}\fi
\ifx \endbibitem  \undefined \def \endbibitem {}\fi
\ifx \bconflocation  \undefined \def \bconflocation#1{#1}\fi
\ifx \arxivurl  \undefined \def \arxivurl#1{\textsf{#1}}\fi
\csname PreBibitemsHook\endcsname

\bibitem[\protect\citeauthoryear{Brunton et~al.}{2016}]{Brunton2016}
\begin{barticle}
\bauthor{\bsnm{Brunton}, \binits{S.L.}},
\bauthor{\bsnm{Proctor}, \binits{J.L.}},
\bauthor{\bsnm{Kutz}, \binits{J.N.}}:
\batitle{Discovering governing equations from data by sparse identification of
  nonlinear dynamical systems}.
\bjtitle{Proc Natl Acad Sci USA}
\bvolume{113}(\bissue{15}),
\bfpage{3932}--\blpage{3937}
(\byear{2016})
\doiurl{10.1073/pnas.1517384113}
\end{barticle}
\endbibitem

\bibitem[\protect\citeauthoryear{Messenger and Bortz}{2021}]{Messenger2021}
\begin{barticle}
\bauthor{\bsnm{Messenger}, \binits{D.A.}},
\bauthor{\bsnm{Bortz}, \binits{D.M.}}:
\batitle{Weak {SINDy} for partial differential equations}.
\bjtitle{J Comput Phys}
\bvolume{443},
\bfpage{110525}
(\byear{2021})
\doiurl{10.1016/j.jcp.2021.110525}
\end{barticle}
\endbibitem

\bibitem[\protect\citeauthoryear{Raissi et~al.}{2019}]{Raissi2019}
\begin{barticle}
\bauthor{\bsnm{Raissi}, \binits{M.}},
\bauthor{\bsnm{Perdikaris}, \binits{P.}},
\bauthor{\bsnm{Karniadakis}, \binits{G.E.}}:
\batitle{Physics-informed neural networks: a deep learning framework for
  solving forward and inverse problems involving nonlinear partial differential
  equations}.
\bjtitle{J Comput Phys}
\bvolume{378},
\bfpage{686}--\blpage{707}
(\byear{2019})
\doiurl{10.1016/j.jcp.2018.10.045}
\end{barticle}
\endbibitem

\bibitem[\protect\citeauthoryear{Brunton et~al.}{2017}]{Brunton2017Havok}
\begin{barticle}
\bauthor{\bsnm{Brunton}, \binits{S.L.}},
\bauthor{\bsnm{Brunton}, \binits{B.W.}},
\bauthor{\bsnm{Proctor}, \binits{J.L.}},
\bauthor{\bsnm{Kaiser}, \binits{E.}},
\bauthor{\bsnm{Kutz}, \binits{J.N.}}:
\batitle{Chaos as an intermittently forced linear system}.
\bjtitle{Nat Commun}
\bvolume{8},
\bfpage{19}
(\byear{2017})
\doiurl{10.1038/s41467-017-00030-8}
\end{barticle}
\endbibitem

\bibitem[\protect\citeauthoryear{Bakarji et~al.}{2023}]{Bakarji2023}
\begin{barticle}
\bauthor{\bsnm{Bakarji}, \binits{J.}},
\bauthor{\bsnm{Champion}, \binits{K.}},
\bauthor{\bsnm{Kutz}, \binits{J.N.}},
\bauthor{\bsnm{Brunton}, \binits{S.L.}}:
\batitle{Discovering governing equations from partial measurements with deep
  delay autoencoders}.
\bjtitle{Proc R Soc A}
\bvolume{479},
\bfpage{20230422}
(\byear{2023})
\doiurl{10.1098/rspa.2023.0422}
\end{barticle}
\endbibitem

\bibitem[\protect\citeauthoryear{Champion et~al.}{2019}]{Champion2019}
\begin{barticle}
\bauthor{\bsnm{Champion}, \binits{K.}},
\bauthor{\bsnm{Lusch}, \binits{B.}},
\bauthor{\bsnm{Kutz}, \binits{J.N.}},
\bauthor{\bsnm{Brunton}, \binits{S.L.}}:
\batitle{Data-driven discovery of coordinates and governing equations}.
\bjtitle{Proc Natl Acad Sci USA}
\bvolume{116}(\bissue{45}),
\bfpage{22445}--\blpage{22451}
(\byear{2019})
\doiurl{10.1073/pnas.1906995116}
\end{barticle}
\endbibitem

\bibitem[\protect\citeauthoryear{Somacal et~al.}{2022}]{Somacal2022}
\begin{barticle}
\bauthor{\bsnm{Somacal}, \binits{A.}},
\bauthor{\bsnm{Barrera}, \binits{Y.}},
\bauthor{\bsnm{Boechi}, \binits{L.}},
\bauthor{\bsnm{Jonckheere}, \binits{M.}},
\bauthor{\bsnm{Lefieux}, \binits{V.}},
\bauthor{\bsnm{Picard}, \binits{D.}},
\bauthor{\bsnm{Smucler}, \binits{E.}}:
\batitle{Uncovering differential equations from data with hidden variables}.
\bjtitle{Phys Rev E}
\bvolume{105}(\bissue{5}),
\bfpage{054209}
(\byear{2022})
\doiurl{10.1103/PhysRevE.105.054209}
\end{barticle}
\endbibitem

\bibitem[\protect\citeauthoryear{Leishman and
  Beddoes}{1989}]{LeishmanBeddoes1989}
\begin{barticle}
\bauthor{\bsnm{Leishman}, \binits{J.G.}},
\bauthor{\bsnm{Beddoes}, \binits{T.S.}}:
\batitle{A semi-empirical model for dynamic stall}.
\bjtitle{J Am Helicopter Soc}
\bvolume{34}(\bissue{3}),
\bfpage{3}--\blpage{17}
(\byear{1989})
\doiurl{10.4050/JAHS.34.3}
\end{barticle}
\endbibitem

\bibitem[\protect\citeauthoryear{Tran and Petot}{1981}]{Tran1980ONERA}
\begin{barticle}
\bauthor{\bsnm{Tran}, \binits{C.T.}},
\bauthor{\bsnm{Petot}, \binits{D.}}:
\batitle{Semi-empirical model for the dynamic stall of airfoils in view of the
  application to the calculation of responses of a helicopter blade in forward
  flight}.
\bjtitle{Vertica}
\bvolume{5}(\bissue{1}),
\bfpage{35}--\blpage{53}
(\byear{1981})
\end{barticle}
\endbibitem

\bibitem[\protect\citeauthoryear{Goman and Khrabrov}{1994}]{GomanKhrabrov1994}
\begin{barticle}
\bauthor{\bsnm{Goman}, \binits{M.}},
\bauthor{\bsnm{Khrabrov}, \binits{A.}}:
\batitle{State-space representation of aerodynamic characteristics of an
  aircraft at high angles of attack}.
\bjtitle{J Aircraft}
\bvolume{31}(\bissue{5}),
\bfpage{1109}--\blpage{1115}
(\byear{1994})
\doiurl{10.2514/3.46618}
\end{barticle}
\endbibitem

\bibitem[\protect\citeauthoryear{Nguyen and Ayoubi}{2027}]{Ayoubi2027}
\begin{bchapter}
\bauthor{\bsnm{Nguyen}, \binits{N.}},
\bauthor{\bsnm{Ayoubi}, \binits{M.A.}}:
\bctitle{Nonlinear unsteady aerodynamic model identification via symbolic
  physics-informed residual augmentation}.
In: \bbtitle{AIAA SciTech 2027 Forum}
(\byear{2027}).
\bcomment{Accepted for presentation}
\end{bchapter}
\endbibitem

\bibitem[\protect\citeauthoryear{Friswell et~al.}{2010}]{Friswell2010}
\begin{bbook}
\bauthor{\bsnm{Friswell}, \binits{M.I.}},
\bauthor{\bsnm{Penny}, \binits{J.E.T.}},
\bauthor{\bsnm{Garvey}, \binits{S.D.}},
\bauthor{\bsnm{Lees}, \binits{A.W.}}:
\bbtitle{Dynamics of Rotating Machines}.
\bpublisher{Cambridge University Press},
\blocation{Cambridge}
(\byear{2010}).
\doiurl{10.1017/CBO9780511780509}
\end{bbook}
\endbibitem

\bibitem[\protect\citeauthoryear{Dunn}{1961}]{Dunn1961}
\begin{barticle}
\bauthor{\bsnm{Dunn}, \binits{O.J.}}:
\batitle{Multiple comparisons among means}.
\bjtitle{J Am Stat Assoc}
\bvolume{56}(\bissue{293}),
\bfpage{52}--\blpage{64}
(\byear{1961})
\doiurl{10.1080/01621459.1961.10482090}
\end{barticle}
\endbibitem

\bibitem[\protect\citeauthoryear{Miller}{1981}]{Miller1981}
\begin{bbook}
\bauthor{\bsnm{Miller}, \binits{R.G.}}:
\bbtitle{Simultaneous Statistical Inference},
\bedition{2nd} edn.
\bpublisher{Springer},
\blocation{New York}
(\byear{1981})
\end{bbook}
\endbibitem

\bibitem[\protect\citeauthoryear{K{\"u}nsch}{1989}]{Kunsch1989}
\begin{barticle}
\bauthor{\bsnm{K{\"u}nsch}, \binits{H.R.}}:
\batitle{The jackknife and the bootstrap for general stationary observations}.
\bjtitle{Ann Stat}
\bvolume{17}(\bissue{3}),
\bfpage{1217}--\blpage{1241}
(\byear{1989})
\doiurl{10.1214/aos/1176347265}
\end{barticle}
\endbibitem

\bibitem[\protect\citeauthoryear{Kingma and Ba}{2015}]{Kingma2015}
\begin{bchapter}
\bauthor{\bsnm{Kingma}, \binits{D.P.}},
\bauthor{\bsnm{Ba}, \binits{J.}}:
\bctitle{{Adam}: a method for stochastic optimization}.
In: \bbtitle{International Conference on Learning Representations (ICLR)}
(\byear{2015}).
\bcomment{arXiv:1412.6980}
\end{bchapter}
\endbibitem

\bibitem[\protect\citeauthoryear{Ljung and Box}{1978}]{LjungBox1978}
\begin{barticle}
\bauthor{\bsnm{Ljung}, \binits{G.M.}},
\bauthor{\bsnm{Box}, \binits{G.E.P.}}:
\batitle{On a measure of lack of fit in time series models}.
\bjtitle{Biometrika}
\bvolume{65}(\bissue{2}),
\bfpage{297}--\blpage{303}
(\byear{1978})
\doiurl{10.1093/biomet/65.2.297}
\end{barticle}
\endbibitem

\bibitem[\protect\citeauthoryear{Halton}{1960}]{Halton1960}
\begin{barticle}
\bauthor{\bsnm{Halton}, \binits{J.H.}}:
\batitle{On the efficiency of certain quasi-random sequences of points in
  evaluating multi-dimensional integrals}.
\bjtitle{Numer Math}
\bvolume{2}(\bissue{1}),
\bfpage{84}--\blpage{90}
(\byear{1960})
\doiurl{10.1007/BF01386213}
\end{barticle}
\endbibitem

\bibitem[\protect\citeauthoryear{Savitzky and Golay}{1964}]{SavitzkyGolay1964}
\begin{barticle}
\bauthor{\bsnm{Savitzky}, \binits{A.}},
\bauthor{\bsnm{Golay}, \binits{M.J.E.}}:
\batitle{Smoothing and differentiation of data by simplified least squares
  procedures}.
\bjtitle{Anal Chem}
\bvolume{36}(\bissue{8}),
\bfpage{1627}--\blpage{1639}
(\byear{1964})
\doiurl{10.1021/ac60214a047}
\end{barticle}
\endbibitem

\bibitem[\protect\citeauthoryear{Montgomery and Runger}{2018}]{montgomery2018}
\begin{bbook}
\bauthor{\bsnm{Montgomery}, \binits{D.C.}},
\bauthor{\bsnm{Runger}, \binits{G.C.}}:
\bbtitle{Applied Statistics and Probability for Engineers},
\bedition{7th} edn.
\bpublisher{John Wiley \& Sons},
\blocation{Hoboken, NJ}
(\byear{2018}).
\bcomment{Chap. 12}
\end{bbook}
\endbibitem

\bibitem[\protect\citeauthoryear{Lehmann and
  Casella}{1998}]{LehmannCasella1998}
\begin{bbook}
\bauthor{\bsnm{Lehmann}, \binits{E.L.}},
\bauthor{\bsnm{Casella}, \binits{G.}}:
\bbtitle{Theory of Point Estimation},
\bedition{2}nd edn.
\bpublisher{Springer},
\blocation{New York, NY}
(\byear{1998})
\end{bbook}
\endbibitem

\bibitem[\protect\citeauthoryear{Ljung}{1999}]{Ljung1999}
\begin{bbook}
\bauthor{\bsnm{Ljung}, \binits{L.}}:
\bbtitle{System Identification: Theory for the User},
\bedition{2nd} edn.
\bpublisher{Prentice Hall},
\blocation{Upper Saddle River, NJ}
(\byear{1999})
\end{bbook}
\endbibitem

\bibitem[\protect\citeauthoryear{S{\"o}derstr{\"o}m and
  Stoica}{1989}]{SoderstromStoica1989}
\begin{bbook}
\bauthor{\bsnm{S{\"o}derstr{\"o}m}, \binits{T.}},
\bauthor{\bsnm{Stoica}, \binits{P.}}:
\bbtitle{System Identification}.
\bpublisher{Prentice Hall},
\blocation{New York}
(\byear{1989})
\end{bbook}
\endbibitem

\bibitem[\protect\citeauthoryear{Tsybakov}{2009}]{Tsybakov2009}
\begin{bbook}
\bauthor{\bsnm{Tsybakov}, \binits{A.B.}}:
\bbtitle{Lower bounds on the minimax risk},
pp. \bfpage{77}--\blpage{135}.
\bpublisher{Springer},
\blocation{New York, NY}
(\byear{2009}).
\doiurl{10.1007/978-0-387-79052-7\_2} .
\burl{https://doi.org/10.1007/978-0-387-79052-7\_2}
\end{bbook}
\endbibitem

\bibitem[\protect\citeauthoryear{Montgomery}{2017}]{Montgomery2017}
\begin{bbook}
\bauthor{\bsnm{Montgomery}, \binits{D.C.}}:
\bbtitle{Design and Analysis of Experiments},
\bedition{9th} edn.
\bpublisher{John Wiley \& Sons},
\blocation{Hoboken}
(\byear{2017})
\end{bbook}
\endbibitem

\bibitem[\protect\citeauthoryear{Piramoon et~al.}{2024}]{Piramoon2024}
\begin{barticle}
\bauthor{\bsnm{Piramoon}, \binits{S.}},
\bauthor{\bsnm{Ayoubi}, \binits{M.A.}},
\bauthor{\bsnm{Bashash}, \binits{S.}}:
\batitle{Modeling and vibration suppression of rotating machines using the
  sparse identification of nonlinear dynamics and terminal sliding mode
  control}.
\bjtitle{IEEE Access}
\bvolume{12},
\bfpage{119272}--\blpage{119291}
(\byear{2024})
\doiurl{10.1109/ACCESS.2024.3449913}
\end{barticle}
\endbibitem

\bibitem[\protect\citeauthoryear{Takens}{1981}]{Takens1981}
\begin{bchapter}
\bauthor{\bsnm{Takens}, \binits{F.}}:
\bctitle{Detecting strange attractors in turbulence}.
In: \beditor{\bsnm{Rand}, \binits{D.}},
\beditor{\bsnm{Young}, \binits{L.S.}} (eds.)
\bbtitle{Dynamical Systems and Turbulence, Warwick 1980}.
\bsertitle{Lecture Notes in Mathematics},
vol. \bseriesno{898},
pp. \bfpage{366}--\blpage{381}.
\bpublisher{Springer},
\blocation{Berlin, Heidelberg}
(\byear{1981}).
\doiurl{10.1007/BFb0091924}
\end{bchapter}
\endbibitem

\end{thebibliography}

\begin{appendices}
\footnotesize

\section{Synthetic 4-DOF rotor matrices}
\label{app:synthetic_matrices}

For the generalized coordinates
$\bq=[x,y,\alpha,\beta]^{\mathsf T}$, the nominal rotor
benchmark (Table~\ref{tab:sim_params}) is governed by
\begin{equation}
\begin{aligned}
\bm M\ddot{\bq}
&+(\bm C_b+\OO\bm G)\dot{\bq}+\bm K_L\bq\\
&+\sum_{i=1}^{2}k_{NLi}\norm{\bL_i\bq}^{2}
  \bL_i^{\mathsf T}\bL_i\bq
=\bm f_{unb}(t).
\end{aligned}
\label{eq:appendix_4dof_model}
\end{equation}
The disk has mass $m=1.15$\,kg, radius $R=0.0735$\,m,
polar inertia $I_p=mR^2/2$, and diametral inertia $I_d=I_p/2$.
The mass and gyroscopic matrices are
\begin{equation}
\begin{aligned}
\bm M&=\operatorname{diag}(m,m,I_d,I_d),\\
\bm G&=
\begin{bmatrix}
0&0&0&0\\
0&0&0&0\\
0&0&0&I_p\\
0&0&-I_p&0
\end{bmatrix}.
\end{aligned}
\label{eq:appendix_MG}
\end{equation}
The bearing span is $1$\,m, and the disk is located $0.25$\,m
from bearing~1. The matrices mapping $\bq$ to the two bearing
deflections are
\begin{equation}
\begin{aligned}
\bL_1&=
\begin{bmatrix}
0.75&0&0&0.1875\\
0&0.75&-0.1875&0
\end{bmatrix},\\
\bL_2&=
\begin{bmatrix}
0.25&0&0&-0.1875\\
0&0.25&0.1875&0
\end{bmatrix}.
\end{aligned}
\label{eq:appendix_Lmatrices}
\end{equation}

Both bearings have identical linear stiffness and damping
matrices. For $i=1,2$,
\begin{equation}
\begin{aligned}
\bm K_{bi}
&\simeq
\begin{bmatrix}
18595.91 & 2542.41\\
2542.41  & 22246.09
\end{bmatrix}
\ \mathrm{N/m},\\
\bm C_{bi}
&=
\begin{bmatrix}
8.0  & 0.16\\
0.16 & 4.8
\end{bmatrix}
\ \mathrm{N\,s/m}.
\end{aligned}
\end{equation}
Projection into the generalized coordinates gives
\begin{equation}
\begin{aligned}
\bm K_L
&=\sum_{i=1}^{2}\bL_i^{\mathsf T}\bm K_{bi}\bL_i\\
&\simeq
\begin{bmatrix}
11622.45 & 1589.01  & -238.35  & 1743.37\\
1589.01  & 13903.81 & -2085.57 & 238.35\\
-238.35  & -2085.57 & 1564.18  & -178.76\\
1743.37  & 238.35   & -178.76  & 1307.53
\end{bmatrix},
\end{aligned}
\end{equation}
\begin{equation}
\begin{aligned}
\bm C_b
&=\sum_{i=1}^{2}\bL_i^{\mathsf T}\bm C_{bi}\bL_i\\
&=
\begin{bmatrix}
5      & 0.1   & -0.015   & 0.75\\
0.1    & 3     & -0.45    & 0.015\\
-0.015 & -0.45 & 0.3375   & -0.01125\\
0.75   & 0.015 & -0.01125 & 0.5625
\end{bmatrix}.
\end{aligned}
\label{eq:appendix_KC}
\end{equation}
The cubic bearing coefficients are
$k_{NL1}=k_{NL2}\simeq3.12195\times10^{11}$\,N/m$^3$.

Let $\psi=\vartheta+\phi_0$, where $\dot{\vartheta}=\OO$,
$\vartheta(0)=0$, and $\phi_0=0.4$\,rad. With static unbalance
$U=me=1.8\times10^{-5}$\,kg\,m, the benchmark uses the force
\begin{equation}
\bm f_{unb}(t)=U
\begin{bmatrix}
-\OO^2\cos\psi-\dot{\OO}\sin\psi\\
-\OO^2\sin\psi+\dot{\OO}\cos\psi\\
0\\
0
\end{bmatrix}.
\label{eq:appendix_unbalance}
\end{equation}
The spin frequency $\OO/(2\pi)$ decreases linearly from 40\,Hz
to zero over 15\,s, giving $\dot{\OO}=-16\pi/3$\,rad/s$^2$.
The output interval is 1\,ms, with four fourth-order Runge--Kutta
substeps per interval. The measurement-noise standard deviation
for each displacement channel is 2\% of the standard deviation
of its noise-free record.

The zero-speed, undamped natural frequencies of the linearized
rotor, obtained from
$\bm K_L\bm v_j=\omega_j^2\bm M\bm v_j$ with
$f_j=\omega_j/(2\pi)$, are
13.78, 16.09, 140.98, and 164.53\,Hz.

\end{appendices}

\end{document}